%% file: main.tex
\documentclass[a4paper,acmsmall,screen]{acmart}
\pdfoutput=1

\usepackage{amsfonts}
\usepackage{stmaryrd}
\usepackage{mathrsfs}
\usepackage{amsmath}
\usepackage{bm}
\usepackage{bbm}
\usepackage{mathtools}
\usepackage{mathpartir}
\usepackage{bbding}
\usepackage{color}
\usepackage[ruled,noend]{algorithm2e}
\usepackage[capitalize,nameinlink]{cleveref}
\usepackage{xparse}
\usepackage{ifthen}
\usepackage{thmtools}
\usepackage{thm-restate}
\usepackage{subcaption}
\usepackage{wrapfig}
\usepackage{tikz}
\usepackage{newfile}
\usepackage{totcount}

\definecolor{blue-green}{rgb}{0,0.5,0.5}

\regtotcounter{@todonotes@numberoftodonotes}

\usepackage{listings}

  \lstdefinestyle{tinyc}{
    basicstyle=\scriptsize\ttfamily,
    keywordstyle=\color{blue}
  }

  \lstdefinestyle{normalc}{
    basicstyle=\ttfamily,
    numbers=none,
    keywordstyle=\color{blue}
  }

  \lstdefinestyle{inlinec}{
    basicstyle=\ttfamily
  }

\usetikzlibrary{backgrounds}
\tikzstyle{every picture}+=[remember picture]

\usepackage[disable]{todonotes}

\newcommand{\VRS}[1]{\todo[inline]{\textcolor{blue}{\textbf{Sathiya}: #1}}}

\usepackage{tikz}
\usetikzlibrary{automata}
\tikzset{gadget/.style={->,>=stealth,initial text=,minimum size=7pt,auto,on grid,scale=1,inner sep=1pt,node distance=1cm}}
\tikzset{every state/.style={minimum size=15pt,inner sep=1pt,fill=black!10,draw=black!70,thick}}

\newbool{inappendix}% Default is \boolfalse{inappendix}
\appto\appendix{\booltrue{inappendix}}% Add boolean switch to \appendix

\usetikzlibrary{backgrounds,calc}

 \usetikzlibrary{ shapes, fit, arrows}
 \usetikzlibrary{decorations} %apparently 'snakes' is deprecated as a tikz library
\DeclareMathSymbol{\mdot}{\mathord}{symbols}{"01}
\usepackage{enumerate}
\usepackage{lscape}
\usepackage{microtype}
\usepackage{enumitem}
\usepackage{tikzit}
\input{tikzit_styles.tikzstyles}
\definecolor{blue-green}{rgb}{0,0.5,0.5}

\input{macros}

\title{Sound and Complete Proof Rules for Probabilistic Termination}

\newcommand{\OurInstitution}{Max Planck Institute for Software Systems (MPI-SWS)}
\newcommand{\OurStreet}{Paul-Ehrlich-Stra{\ss}e, Building G26}
\newcommand{\OurCity}{Kaiserslautern}
\newcommand{\OurPostcode}{67663}
\newcommand{\OurCountry}{Germany}

\author{Rupak Majumdar}
\orcid{0000-0003-2136-0542}             %% \orcid is optional
\affiliation{
  \institution{\OurInstitution}            %% \institution is required
  \streetaddress{\OurStreet}
  \city{\OurCity}
  \postcode{\OurPostcode}
  \country{\OurCountry}                    %% \country is recommended
}

\email{rupak@mpi-sws.org}          %% \email is recommended

\author{V. R. Sathiyanarayana}
\orcid{0009-0006-5187-5415}             %% \orcid is optional
\affiliation{
  \institution{\OurInstitution}            %% \institution is required
  \streetaddress{\OurStreet}
  \city{\OurCity}
  \postcode{\OurPostcode}
  \country{\OurCountry}                    %% \country is recommended
}
\email{sramesh@mpi-sws.org}          %% \email is recommended

\input{abstract}

\begin{document}

\sloppy

\makeatletter
\newtheoremstyle{dazzle}%
{.5\baselineskip\@plus.2\baselineskip
  \@minus.2\baselineskip}% space above
{.5\baselineskip\@plus.2\baselineskip
  \@minus.2\baselineskip}% space below
{\@acmplainbodyfont}% body font
{\@acmplainindent}% indent amount
{\bfseries}% head font
{.}% punctuation after head
{.5em}% spacing after head
{\thmname{\textcolor{red}{\textbf{#1}}}\thmnumber{ \textcolor{red}{\textbf{#2}}}\thmnote{ {\@acmplainnotefont(\textcolor{blue}{#3})}}}% head spec
\makeatother

\theoremstyle{dazzle}
\newtheorem{maintheorem}[theorem]{Theorem}
\newtheorem{mainlemma}[theorem]{Lemma}
\newtheorem{maincorollary}[theorem]{Corollary}
\newtheorem{mainproposition}[theorem]{Proposition}

\crefalias{maintheorem}{theorem}
\crefalias{mainlemma}{lemma}
\crefalias{maincorollary}{corollary}
\crefalias{mainproposition}{proposition}

\theoremstyle{acmplain}
\newtheorem{observation}[theorem]{Observation}
\theoremstyle{acmdefinition}
\newtheorem{remark}[theorem]{Remark}

%\crefname{maintheorem}{Theorem}{Theorems}
%\crefname{mainlemma}{Lemma}{Lemmas}
%\crefname{maincorollary}{Corollary}{Corollaries}
%\crefname{mainproposition}{Proposition}{Propositions}
%\crefname{observation}{Observation}{Observations}
%\crefname{remark}{Remark}{Remarks}
%\Crefname{maintheorem}{Theorem}{Theorems}
%\Crefname{mainlemma}{Lemma}{Lemmas}
%\Crefname{maincorollary}{Corollary}{Corollaries}
%\Crefname{mainproposition}{Proposition}{Propositions}
\Crefname{observation}{Observation}{Observations}
%\Crefname{remark}{Remark}{Remarks}

% Mark the code listings as "Code" instead of "Listings"
\renewcommand{\lstlistingname}{\listingsInLatex}

\maketitle

\input{intro}

\input{prelims}

\input{ast-martingale}

\input{quant}

\input{applications}

% \RM{THE REST IS CANNIBALIZED REMAINS OF THE PREV TEXT}

% \input{previous_rules}

% \input{prev_limitations}

% \input{our_rules}

% \input{related}

% \input{conclusion}

% \begin{acks}
% % We thank the reviewers for their helpful comments.
% This research was sponsored in part by
% the Deutsche Forschungsgemeinschaft project 389792660 TRR 248--CPEC
% (see \url{https://perspicuous-computing.science}).
% \end{acks}

\sloppy % line breaks

\newpage

% \bibliographystyle{ACM-Reference-Format}
% Write total number of pages into file main.pages.ctr
\label{beforebibliography}
\newoutputstream{pages}
\openoutputfile{main.pages.ctr}{pages}
\addtostream{pages}{\getpagerefnumber{beforebibliography}}
\closeoutputstream{pages}
\bibliography{bibliography}

% Write number of pages of main part into file main.pages.ctr
\label{afterbibliography}
\newoutputstream{pagesbib}
\openoutputfile{main.pagesbib.ctr}{pagesbib}
\addtostream{pagesbib}{\getpagerefnumber{afterbibliography}}
\closeoutputstream{pagesbib}

\appendix

\input{appendix}

% \input{discuss}

% Write number of todo notes into file main.todos.ctr
\newoutputstream{todos}
\openoutputfile{main.todos.ctr}{todos}
\addtostream{todos}{\arabic{@todonotes@numberoftodonotes}}
\closeoutputstream{todos}

% Write total number of pages into file main.pagestotal.ctr
\label{endofdocument}
\newoutputstream{pagestotal}
\openoutputfile{main.pagestotal.ctr}{pagestotal}
\addtostream{pagestotal}{\getpagerefnumber{endofdocument}}
\closeoutputstream{pagestotal}

\end{document}

%% file: tikzit_styles.tikzstyles
\tikzstyle{plain circle}=[fill=white, draw=black, shape=circle]
\tikzstyle{black circle}=[fill=black, draw=black, shape=circle]
\tikzstyle{blue circle}=[fill=blue, draw=black, shape=circle]
\tikzstyle{red circle}=[fill=red, draw=black, shape=circle]
\tikzstyle{blue-green circle}=[fill={rgb,255: red,0; green,128; blue,128}, draw=black, shape=circle]

\tikzstyle{direction edge}=[->]
\tikzstyle{dashed directional edge}=[->, dashed]
\tikzstyle{red full edge}=[draw=red, ->]
\tikzstyle{dashed simple edge}=[-, dashed]
\tikzstyle{double directional}=[double, ->]
\tikzstyle{dash double}=[dashed, double, ->]
\tikzstyle{new edge style 0}=[-, dashed, draw=red]
\tikzstyle{blue}=[-, draw=blue]
\tikzstyle{blue directiona}=[draw=blue, ->]
\tikzstyle{blue dashed}=[dashed, draw=blue, ->]

%% file: macros.tex
\def\nat{\mathbb{N}}

\def\cG{\mathcal{G}}
\def\cE{\mathcal{E}}
\def\cH{\mathcal{H}}
\def\cU{\mathcal{U}}
\def\cK{\mathcal{K}}
\def\cF{\mathcal{F}}
\def\mbf#1{\mathbf{#1}}
\def\bfzero{\mbf{0}}
\def\bfx{\mbf{x}}

\def\frakf{\mathfrak{f}}
\def\bbE{\mathbb{E}}

\def\relshortuparrow{\mathrel{\shortuparrow}}
\def\cG{\mathcal{G}}%grammars
\DeclareDocumentCommand{\autstep}{O{}}{%
        \xrightarrow{#1}%
        }
\DeclareDocumentCommand{\langof}{O{} m}{%
  \mathsf{L}_{#1}(#2)%
  }

\newcommand{\pGCL}{\mathsf{pGCL}}
\newcommand{\pChoice}[1]{\oplus_{#1}}
\newcommand{\pChoiceWithoutParamters}{\oplus}

\newcommand{\codeStyleMath}[1]{\mathtt{#1}}

\newcommand{\ttx}{\codeStyleMath{x}}

\newcommand{\iCFG}{\mathsf{CFG}}
\newcommand{\CFG}{\mathsf{CFG}}
\newcommand{\sfUp}{\mathsf{Upd}}
\newcommand{\sfPr}{\mathsf{Pr}}

\newcommand{\ordMapsTo}{\mathord{\mapsto}}
\newcommand{\AST}{\mathsf{AST}}
\newcommand{\PAST}{\mathsf{PAST}}
\newcommand{\BAST}{\mathsf{BAST}}
\newcommand{\sched}{\mathfrak{s}}
\newcommand{\runs}{\mathsf{Runs}}
\newcommand{\sfTerm}{\mathsf{Term}}
\newcommand{\sfLeave}{\mathsf{Leave}}
\newcommand{\sfEnum}{\mathsf{Enum}}

\newcommand{\Prob}{\mathbb{P}}
\newcommand{\cyl}{\mathsf{Cyl}}
\newcommand{\init}{\mathit{init}}
\newcommand{\out}{\mathit{out}}

\newcommand{\reach}{\mathsf{Reach}}
\newcommand{\Inv}{\mathsf{Inv}}
\newcommand{\Nat}{\mathsf{Nat}}
\newcommand{\Lo}{\mathsf{Lo}}
\newcommand{\Hi}{\mathsf{Hi}}
\newcommand{\SI}{\mathsf{SI}}
\newcommand{\SIpred}{\mathbf{\Psi}}

\newcommand{\ThN}{\mathsf{Th}(\mathbb{N})}
\newcommand{\ThQ}{\mathsf{Th}(\mathbb{Q})}
\newcommand{\intQ}{\mathsf{I}_\mathbb{Q}}
\newcommand{\intN}{\mathsf{I}_\mathbb{N}}

\newlist{inlinelist}{enumerate*}{1}
\setlist*[inlinelist,1]{%
  label=(\alph*),
}

\newcommand{\termProb}{\Pr_{\mathrm{term}}}

\newcommand{\setOfReals}{\mathbb{R}}

\newcommand{\setOfPositiveReals}{\mathbb{R}^+}
\newcommand{\setOfRationals}{\mathbb{Q}}

\newcommand{\setOfNaturals}{\mathbb{N}}
\newcommand{\setOfIntegers}{\mathbb{Z}}

\newcommand{\listingsInLatex}{Code}

\newcommand{\listingLineRef}[1]{Line ~\ref{#1}}

\newcommand{\codeRef}[1]{Program \ref{#1}} % For figures
\newcounter{codeCounter}
\newcounter{tempFigureCounter}
\newcommand{\declareCodeFigure}{
  \renewcommand\figurename{Prg.}
  \setcounter{tempFigureCounter}{\value{figure}}
  \setcounter{figure}{\value{codeCounter}}
  \addtocounter{codeCounter}{1}
} % Use inside fig. defn for code figs
\newcommand{\finishcodefigure}{
  \setcounter{figure}{\value{tempFigureCounter}}
}

\newcommand{\node}{\mathsf{node}}

\makeatletter
\newtheoremstyle{dazzle}%
{.5\baselineskip\@plus.2\baselineskip
  \@minus.2\baselineskip}% space above
{.5\baselineskip\@plus.2\baselineskip
  \@minus.2\baselineskip}% space below
{\@acmplainbodyfont}% body font
{\@acmplainindent}% indent amount
{\bfseries}% head font
{.}% punctuation after head
{.5em}% spacing after head
{\thmname{\textcolor{red}{\textbf{#1}}}\thmnumber{ \textcolor{red}{\textbf{#2}}}\thmnote{ {\@acmplainnotefont(\textcolor{blue}{#3})}}}% head spec
\makeatother

\usepackage[framemethod=TikZ]{mdframed} % Add nobreak to ensure no page break
\newcounter{proofrule}[section]
\renewcommand{\theproofrule}{{\ifinappendix\Alph{section}\else\arabic{section}\fi}.\arabic{proofrule}}
\newenvironment{proofrule}[1][]{%
\refstepcounter{proofrule}%
\smallskip
\ifstrempty{#1}%
{\mdfsetup{%
frametitle={%
\tikz[baseline=(current bounding box.east),outer sep=0pt]
\node[anchor=east,rectangle,fill=blue!20]
{\strut Proof Rule~\theproofrule};}}
}%
{\mdfsetup{%
frametitle={%
\tikz[baseline=(current bounding box.east),outer sep=0pt]
\node[anchor=east,rectangle,fill=blue!20]
{\strut Proof Rule~\theproofrule:~#1};}}%
}%
\mdfsetup{innertopmargin=3pt,linecolor=blue!20,%
skipbelow=7pt, skipabove=4pt%
linewidth=2pt,topline=true,%
outerlinewidth=1pt,
}
\begin{mdframed}[]\relax%
}{\end{mdframed}}

\crefname{proofrule}{Rule}{Rules}
\Crefname{proofrule}{Rule}{Rules}

%% file: abstract.tex
\begin{abstract}
Deciding termination is a fundamental problem in the analysis of probabilistic imperative programs.
We consider the \emph{qualitative} and \emph{quantitative} probabilistic termination problems for an imperative
programming model with discrete probabilistic choice and demonic bounded nondeterminism.
The qualitative question asks if the program terminates almost-surely, no matter how nondeterminism is resolved.
The quantitative question asks for a bound on the probability of termination.
Despite a long and rich literature on the topic, no sound and relatively complete proof systems were known for these problems.
In this paper, we provide such sound and relatively complete proof rules for proving qualitative and quantitative termination in the assertion language of arithmetic.
Our rules use supermartingales as estimates of the likelihood of a program's evolution and variants as measures of distances to termination.
Our key insight is our completeness result, which
shows how to construct a suitable supermartingales from an almost-surely terminating program.
We also show that proofs of termination in many existing proof systems can be transformed to proofs in our system, pointing
to its applicability in practice.
As an application of our proof rule, we show an explicit proof of almost-sure termination for the two-dimensional random walker.
\end{abstract}

%% file: intro.tex
\section{Introduction}

Probabilistic programming languages extend the syntax of the usual deterministic models of computation with primitives for random choice.
In this way, probabilistic programs express randomized computation and have found applications in many domains where randomization is essential.

We study the \emph{termination problem} for probabilistic programs with discrete probabilistic and nondeterministic choice.
Termination is a fundamental property of programs and formal reasoning about (deterministic) program termination goes back to Turing \cite{Turing37}.
Its extension to the probabilistic setting can be either qualitative or quantitative.
\emph{Qualitative} termination, or Almost-Sure Termination $(\AST)$, asks if the program terminates with probability $1$, no matter how the nondeterminism is resolved.
\emph{Quantitative} termination, on the other hand, relates to finding upper and lower bounds on the probability of termination that hold across all resolutions of nondeterminism.

For finite-state probabilistic programs, both qualitative and quantitative termination problems are well understood:
there are sound and complete \emph{algorithmic} procedures for termination that operate by analyzing the underlying finite-state Markov decision processes.
Intuitively, every run of the system eventually arrives in an end component (a generalization of bottom strongly connected components to Markov decision processes) and computing the termination probabilities consequently reduces to computing the reachability probabilities for the
appropriate end components \cite{HartSP83,Vardi85,CourcoubetisY95,BdA95,dAHK07,dAH00}.

% I think the following paragraph is a bit misleading; from my reading, Katoen et. al. promise a relatively complete assertion language, and say absolutely nothing about proof rules.
% For infinite state spaces, existing sound and relatively complete proof rules are \emph{extensional}: they
% depend on semantic characterizations of pre-expectations \cite{HartS85,BatzKKM21,HartogV02}.
% They show, in essence, that an expressive expression language---either one of functions using arithmetic with 
% supremum and infimum, or just arithmetic with greatest and least fixed point operators---can express pre-expectations of programs.
% These rules satisfy the letter but not the spirit of proof rules: they show that a sufficiently strong assertion language
% can express all probabilistic behaviors but do not show how to break down the semantics into a sequence of certificates and checks in practice.
% In contrast, practical attempts to find proof rules for probabilistic termination are \emph{intensional}: they aim for explicit certificates
% that break down the reasoning.
% In the absence of relatively complete intensional proof systems, existing approaches try to extend the scope or automation
% of previous rules, but do not provide a guarantee of (relative) completeness \cite{HartSP83,dAHK07,McIverMorganBook,ChakarovSankaranarayanan,ChatterjeeNZ17,McIverMKK18,ChatterjeeGMZ22}.

The story is different for infinite state spaces.
Existing techniques for deducing termination typically take the form of proof rules over a program logic \cite{HartSP83,dAHK07,McIverMorganBook,ChakarovSankaranarayanan,ChatterjeeNZ17,McIverMKK18,ChatterjeeGMZ22}.
These rules ask for certificates consisting of a variety of mathematical entities that satisfy locally-checkable properties.
\emph{None of them, however, are known to be complete}; meaning we do not know if certificates matching these rules can always be found for terminating programs.
The search for sound and complete proof rules has been a long-standing open problem.
Note that the undecidability of the termination problem means that one cannot hope for complete algorithms for generating such certificates.
% and more many of these, there exist terminating programs that do not yield certificates matching the rule.
% Furthermore, many of them are written for quantitative logics that detail no formal language for specifying assertions.
% Instead, they rely on general mathematical denotations, which are often too powerful.
% While applicable assertion languages \cite{BatzKKM21} exist in the literature, to the best of our knowledge no effort has been made translating existing rules into these formal languages.

In this paper, we provide sound and relatively complete proof rules for qualitative and quantitative termination of probabilistic programs.
We present our rules in a simple proof system in the style of \citet{Floyd1993} that applies naturally to our program model.
The underlying assertion language of the proof system is arithmetic interpreted over the standard model of rational numbers. 
Soundness means that if our proof rule applies, then the system satisfies the (qualitative or quantitative) termination criterion.
Completeness of our rules is relative to the completeness of a proof system for the underlying assertion language of arithmetic. 
Accordingly, we show an effective reduction from the validity of our program logic to the validity of a finite number of assertions 
in arithmetic.
For completeness, whenever the original program terminates (qualitatively or quantitatively), we can construct a proof in our program logic in 
such a way that all relevant certificates can be expressed in the assertion language.
This is important: merely knowing that certain semantic certificates exist may not be sufficient for a proof system, 
e.g., if these certificates are provided non-constructively or require terms that cannot be expressed in the assertion language.

Let us now be more precise.
We work in an imperative programming model with variables ranging over the rationals.
Our model fixes a finite set of program locations, and defines a guarded transition relation between the locations representing computational steps.
At marked locations, the model contains primitives for probability distributions over available transitions.
This allows for the expression of \emph{bounded} nondeterministic and probabilistic choice;
% In addition to assignments, conditionals, and loops, our language contains a \emph{probabilistic choice} operator that picks
% successor program blocks according to a (discrete) distribution, and a \emph{binary nondeterministic choice} operator that picks one of two
% successor blocks nondeterministically.
we assume the nondeterminism is resolved demonically.
We fix the language of arithmetic as our expression and assertion language, and interpret formulas over the standard model of the rational numbers.\footnote{
	One can generalize our result to \emph{arithmetical} structures \cite{HarelKozenTiuryn}, but we stick to the simpler setting for clarity.
}
The semantics of our programming language is given by a Markov decision process on countably many states, where a demonic scheduler resolves the nondeterminism.
Since the language has bounded nondeterminism, we note that each state has a finite number of immediate successor states.
%% RM: this part is too technical here:
%% and so, for every scheduler, the number of states reached in a bounded number of steps is finite.

Given a program and a terminal state, the \emph{qualitative termination} question asks if the infimum of the probability of reaching the terminal state across all resolutions of nondeterminism is one.
In other words, it asks if the program almost-surely terminates under all possible schedulers.
The \emph{quantitative termination} question asks if the probability of reaching the terminal state is bounded above or below by a given probability value $p$.

For the special case of programs without probabilistic choice, sound and relatively complete proof systems for termination are known \cite{MannaP74,Harel80,Apt81}:
they involve finding a \emph{variant} function from (reachable) program states to a well-founded domain that decreases on every step of the program
and maps the terminal state to a minimal element. 

A natural generalization of variant functions to the probabilistic setting is the \emph{ranking supermartingale}: a function from states to reals that reduces in expectation by some amount on each step of the program.
Ranking supermartingales of various flavors are the workhorse of existing proof rules for qualitative termination.
Unfortunately, an example from \cite{MS24} shows that a proof rule based only on a ranking supermartingale is incomplete: there may not exist a ranking
supermartingale that decreases in expectation on each step and one may require transfinite ordinals in proofs.
% SATHIYA: This example seems overkill, as the 1-D random walk should suffice to demonstrate the insufficiency of ranking supermartingales. Admittedly, I don't know if anyone had specifically showed this earlier, so I haven't acted on this now.

% \subsubsection*{Basic Ingredients}
We take a different perspective.
Instead of asking for a representation of the distance to a terminal state in line with prior techniques, we model the relative likelihood of the program's evolution and deal with distances to termination separately.
% We provide two different proof rules for qualitative termination.
In the qualitative setting, we use a function that is unbounded and non-increasing in expectation to certify almost-sure termination.
In the quantitative case, we certify upper and lower bounds on the probability of termination using a family of bounded functions, each non-increasing in expectation.
% In the second, we certify almost-sure termination using a family of bounded functions, one for each reachable state, each non-increasing in expectation.
Both certificates track the execution's relative likelihood in subtly different ways, and both require additional side conditions.

% Our proofs are built out of several simple observations.
% Second, if the probability of termination from a state $x$ is $1$, and $x$ can reach $y$, then the probability of termination from $y$ is also $1$. 
% Third, the probability of an event is at least $p$ iff it is at least $p - \epsilon$ for every $\epsilon > 0$.

% We prove an \emph{unrolling lemma} that is a central tool for our completeness results. 
% It states that if the infimum over all schedulers of the probability of reaching a terminal state is at least $p$, then 
% for every $\epsilon$, there is a finite upper bound $k$ such that the infimum over all schedulers of the probability mass of reaching the terminal state
% within $k$ steps is at least $p - \epsilon$.
% In particular, the set of states reachable in $k$ steps defines a finite state space.
% The unrolling lemma appears as a basic ingredient in characterizing the complexity of almost sure termination \cite{KaminskiKM19,MS24arxiv}.
% We show that it provides a surprisingly powerful tool in proving completeness of proof systems by 
% ``carving out'' finite state systems out of infinite-state termination problems.

\subsubsection*{Qualitative Termination}
Our rule asks for a supermartingale \cite{doobBook} function $V$ that is non-increasing in expectation on all states
and a variant function $U$ that certifies that every reachable state has some finite path to the terminal state.
The rule also mandates certain compatibility criteria on $U$ and $V$ to enable the execution's almost-sure either termination or escape from state spaces within which $V$ is bounded.

We show the rule is sound by appropriately partitioning the collection of all possible executions and strategically employing variant arguments and/or martingale theory within each partition.
Its completeness uses the following observation.
% Let the finite set be the (wlog, singleton) terminal state, and the inductive invariant be the set of reachable states.
Fix an enumeration of the reachable states $s_1, s_2, \ldots$ of an almost-surely terminating program.
Let $\Pr_s[\Diamond \mathbbm{1}_{> n}]$ denote the probability with which a run starting from a state $s$ reaches some state in $\{s_n, s_{n+1}, \ldots\}$ in the enumeration.
For a fixed $s$, we show that $\Pr_s[\Diamond \mathbbm{1}_{> n}]$ goes to zero as $n\rightarrow \infty$.
Following a diagonal-like construction from the theory of countable Markov chains \cite{MSZ78}, we define a sequence $n_1, n_2, \ldots$ such 
that $\Pr_j [\Diamond \mathbbm{1}_{s_{n_k}}] \leq \frac{1}{2^k}$ for all $j\leq k$; this is possible because the aforementioned limit is $0$.
This lets us define a supermartingale defined over the states $s$ as
\[
\sum_{k\in \nat} \Pr{}_s [\Diamond \mathbbm{1}_{> n_k}]
\]
This supermartingale satisfies the requirements of our rule.
Moreover, we show that this supermartingale can be formally expressed in first-order arithmetic.
Thus, the rule is relatively complete.

\subsubsection*{Quantitative Termination}
Our rules for quantitative termination employ \emph{stochastic invariants} \cite{ChatterjeeNZ17}: pairs $(\SI, p)$ such that the probability with which executions leave the set of states $\SI$ is bounded above by $p$.
A stochastic invariant easily implies an upper bound rule: if there is a stochastic invariant $(\SI, p)$ that avoids terminal
states, the probability of termination is upper bounded by $p$.
For lower bounds, our starting point is the proof rule proposed by \citet{ChatterjeeGMZ22} using stochastic invariants:
if there is a stochastic invariant $(\SI, p)$ such that runs almost-surely terminate within $\SI$ or leave $\SI$, then
the probability of termination is at least $1-p$.
While they claimed soundness and completeness for their rule, there were two issues.
First, completeness of their rule assumed that there exists a complete proof rule for qualitative termination.
However, no relatively complete proof rule for qualitative termination was known before.
Second, we show in \cref{sec:ctrex} an explicit example where their rule, as stated, is incorrect.

In this paper, we show a sound and complete rule that is a modification of their rule:
we require that for each $n\in \nat$, there is a stochastic invariant $\left(\SI_n, p+ \frac{1}{n}\right)$ such that all runs almost-surely
terminate within $\SI_n$ or leave $\SI_n$.
Sound and relatively complete certificates for almost-sure termination can now be given using a sound and complete proof rule for qualitative termination for finite-state systems.
Interestingly, the quantitative rules lead, as a special case, to a different sound and complete proof rule for qualitative termination.
We can show that proofs using our qualitative termination rule can be compiled into proofs using this new rule and vice versa.

\subsubsection*{Other Related Work}
The supermartingales used by our proof rules are closely related to Lyapunov functions used in the study of the stability of dynamical systems.
Lyapunov functions are classical tools for characterizing recurrence and transience in infinite-state Markov chains, going back to the work of Foster \cite{Foster51,Foster53}.
Completeness of Lyapunov functions was shown in general by \citet{MSZ78}.
Our proofs of soundness and completeness uses insights from \citet{MSZ78}, but we have to overcome several technical issues.
First, we have demonic nondeterminism and therefore require machinery for reasoning about infimums over all schedulers.
Second, we do not have irreducibility: while an irreducible Markov chain is either recurrent or transient, we cannot assume that a program that is not almost-sure terminating has a strong transience property.
Thus, we have to prove these properties ab initio. 

A side effect of our completeness results is that certain existing proof rules can be shown to be complete by reducing them to our rules.
% In the literature, there exist sound proof rules for $\AST$ that use supermartingales.
% One rule by \citet{HuangFC18} uses supermartingales that exhibit a lower bound on their variation in each step.
One such rule for almost-sure termination was presented \citet{McIverMKK18}.
Their proof rule asks for a supermartingale function that also acts as a distance variant, and is sound for almost-sure termination.
This rule was believed to be incomplete, even by the authors \cite{katoenETAPSTalk}, despite the fact that no counterexamples against completeness could be found. % NOTE: In the link to the talk, (https://www.youtube.com/watch?v=I3sOp_mbs8k), Katoen says, "I can't imagine that this (rule) is complete" at 1:30:22
By extending the techniques underlying the completeness of our rules, we prove that, surprisingly, their rule is also relatively complete.
% However, this completeness doesn't induce natural techniques for building certificates.
% \VRS{I want to argue here: why is our rule better than theirs, given how similar they are?}
However, we argue that, when compared to their rule, our rules generally lead to more explainable and comprehensible certificates.
This is because the variant function that our rules ask for match the intuitive notion of ``distance'' to termination.
Additionally requiring the variant to be a supermartingale places unnecessary constraints on it that don't yield a better understanding of the program's termination.
We illustrate this with a simple example in \cref{subsec:mciver-katoen-rule}.

The issue of an appropriate assertion language for proof rules for termination have been mostly elided in the literature, and most rules are presented in an informal
language of ``sufficiently expressive'' mathematical constructs.
Important exceptions are the assertion languages of \citet{BatzKKM21} and \citet{HartogV02}.
Their work shows that the language of arithmetic extended with suprema and infima of
functions over the state space is relatively complete (in the sense of \citet{Cook78}) for weakest-preexpectation style reasoning of probabilistic programs without nondeterminism.
That is, given a function $f$ definable in their language and a program $P$, they show that their language is expressive enough to represent the weakest pre-expectation of $f$ with respect to $P$.
The need for suprema and infima is motivated by demonstrating simple probabilistic programs whose termination probabilities involve transcendental numbers.
% This language (and the language of \citet{HartogV02} in the same direction) do not account for bounded nondeterministic choice.
% The proof rules presented in this work can be modified in a way that the assertions they require can be written in this language instead.
%
We note that arithmetic is sufficient for relative completeness because suprema and infima arising in probabilistic termination can be encoded (through quantifiers).
We believe that presenting the rules directly in the language of arithmetic allows greater focus on the nature of the certificates required by the rules.
Note that for extensions of the programming model, such as with unbounded nondeterministic choice, arithmetic is no longer sufficient for relative completeness,
and this holds already without probabilistic choice in the language \cite{HitchcockP72,AptK86,AptP86}.

While we focus on almost-sure termination, there are related qualitative termination problems: \emph{positive} almost-sure termination ($\PAST$) and \emph{bounded} almost-sure termination $(\BAST$).
These problems strengthen almost-sure termination by requiring that the expected time to termination is finite.
(Note that a program may be almost-surely terminating but the expected run time may be infinite: consider a one-dimensional symmetric random walk where $0$ is an absorbing state.)
The difference is that $\PAST$ allows the expected run time to depend on the scheduler that resolves nondeterminism, and $\BAST$ requires a global bound that holds for every scheduler.
Sound and complete proof rules for $\BAST$ have been studied extensively \cite{BournezG05,FioritiH15,FuC19,AvanziniLY20}.
More recently, a sound and complete proof rule for $\PAST$ was given \cite{MS24}.
Completeness in these papers are semantic, and relative completeness in the sense of Cook was not studied.
Our techniques would provide a relative completeness result for $\BAST$.
In contrast, \citet{MS24} show that $\PAST$ is $\Pi_1^1$-complete (AST and $\BAST$ are arithmetical, in comparison); thus arithmetic is no longer sufficient as an
assertion language for relative completeness (see \cite{AptP86} for similar issues and appropriate assertion languages for nondeterministic programs with countable nondeterminism).

Our results apply to discrete probabilistic choice.
While discrete choice and computation captures many randomized algorithms, our proofs of completeness do not apply to programs with,
e.g., sampling from continuous probability distributions.
The use of real values introduce measure-theoretic complexities in the semantics \cite{TakisakaOUH21}.
These can be overcome, but whether there is a sound and relatively complete proof rule for an appropriate assertion language remains open.

While we focus on the theoretical aspects here, there is a large body of work on tools for \emph{synthesizing} certificates automatically for probabilistic verification and general termination
\cite{ChakarovSankaranarayanan,ChatterjeeGMZ22,McIverMKK18,FengCSKKZ23,HuangFC18,Monniaux01,KobayashiLG20,BartheEFH16,KincaidPaper,PiskacW11}.
We show a ``compilation'' of many existing rules into our rules, thus, such tools continue to work in our proof system.

\smallskip
\noindent\emph{Contributions.}
We now summarize our main contributions:
\begin{enumerate}
	\item We provide the \emph{first} sound and relatively complete proof rules for both the qualitative and quantitative termination in the assertion language of arithmetic. 
	This solves a 40-year old open question in the field and culminates the substantial body of work on proof systems for probabilistic termination. 

	\item We demonstrate the applicability of our rules by reducing existing proof rules in the literature to our rules.
	In addition, we show that McIver et al.'s \cite{McIverMKK18}  prior proof rule for qualitative termination is complete: completeness was open and the rule was conjectured to be incomplete.

	\item As a particular application, we provide an explicit proof for the almost-sure termination of the 2-D Random Walker within our rules for qualitative termination.
	An explicit supermartingale-based proof for this problem stumped prior techniques (see, e.g., \cite{McIverMKK18}). % I really like the word "stumped" XD
\end{enumerate}

% In summary, we provide the first both sound and relatively complete proof rules for qualitative and quantitative termination, culminating the substantial body of work on probabilistic termination in the last four decades.
% Our rules use familiar ingredients and emphasizes the power of appropriately defined finite-state systems hiding inside the infinite state verification problem.
% We consider these as positive features---while our rules are simple, they provide natural sound and relatively complete proof systems for
% almost sure termination, and were missed in 

%% file: prelims.tex
\section{Probabilistic Programs}

\subsection{Syntax and Semantics}
\label{subsec:syntax-semantics}

\subsubsection*{Syntax.} % Switching out the program variables
We work with \emph{probabilistic control flow graphs}, a program model used by \citet{ChatterjeeGMZ22} to detail proof rules for quantitative termination.
Program variables in this model range over the rationals.
Assignment statements and loop guards are terms and boolean combinations of atomic formulae expressible in the language of arithmetic.
This is standard in program logics \cite{Apt81}, and facilitates the use of the language of rational arithmetic with addition, multiplication, and order to make assertions about desirable program properties.
% For sake of conciseness, we augment this assertion language with additional computable predicates as ``syntactic sugar'' in our proofs. %TODO: We should move this bit later, I think
% We interpret assertions over the standard model of rationals. %TODO: Consider including a statement on the ability to simulate the theory of naturals in this using Julia Robinson's paper

\begin{definition}[Control Flow Graphs]
    A Control Flow Graph ($\iCFG$) $\cG$ is a tuple $(L, V, l_{init}, \bfx_{init}, \ordMapsTo, G, \sfPr, \sfUp)$, where
    \begin{itemize}
        \item $L$ is a finite set of program locations, partitioned into assignment, nondeterministic, and probabilistic locations $L_A$, $L_N$, and $L_P$, respectively.

        \item $V = \{ x_1, x_2, \ldots x_n \}$ is a finite set of program variables.

        \item $l_{init} \in L$ is the initial program location, and $\bfx_{init} \in \setOfRationals^{V}$ is the initial variable valuation.

        \item $\ordMapsTo \subseteq L \times L$ is a finite set of transitions. If $\tau = (l, l') \in \ordMapsTo$, then $l$ and $l'$ are respectively referred to as the source and target locations of $\tau$.

        \item $G$ is a function mapping each $\tau \in \ordMapsTo$ to a Boolean expression over $V \cup L$.

        \item $\sfPr$ is a function assigning each $(l, l') \in \ordMapsTo$ with $l \in L_P$ a fractional expression $p$ over the variables $V$ representing a rational probability value. 

        \item $\sfUp$ is a map assigning each $(l, l') \in \ordMapsTo$ with $l \in L_A$ an update pair $(j, u)$ where $j \in \{1, \ldots, |V|\}$ is the target variable index and $u$ is an arithmetic expression over the variables $V$.

        \item At assignment locations, there is at most one outgoing transition.
        
        \item At probabilistic locations $l$, it must be that $\sfPr(l, \_)[\bfx] > 0$ and $\sum G(l, \_)[(l, \bfx)] \times \sfPr((l, \_))[\bfx] = 1$ over all transitions $(l, \_) \in \ordMapsTo$ for all $\bfx \in \setOfRationals^{V}$.
    \end{itemize}
\end{definition}
% The domains of $G$, $\sfPr$, and $\sfUp$ are considered to be $L \times L$.
We use the boldface notation $\bfx$ for variable assignments and write $\bfzero$ for the assignment that maps every variable to zero.
$\sfPr(l, l')[\bfx]$, $G(l, l')[\bfx]$, and $\sfUp(l, l')[\bfx]$ refer to the output of the expressions $\sfPr(l, l')$, $G(l, l')$, and $\sfUp(l, l')$ on the assignment $\bfx$.
Note that the finiteness of $L$ implies that the branching at both nondeterministic and probabilistic locations is bounded.
Without loss of generality, we assume simple structural conditions that ensures that every state has a successor.
This follows similar assumptions made in prior work \cite{ChatterjeeGMZ22}.
Observe that, while the probabilistic choice is simple, it is sufficient to model probabilistic Turing machines and some
quite sophisticated probabilistic phenomena \cite{FlajoletPS11}.
However, we explicitly forbid unbounded nondeterministic choice or sampling from continuous distributions.

\begin{remark}
While we use $\iCFG$s as our formal model of programs, we could have equivalently used probabilistic guarded command language ($\pGCL$).
$\pGCL$ is the probabilistic extension of the Guarded Command Language of \citet{Dijkstra76}, and is a convenient language for specifying probabilistic computation.
There is a large body of work \cite{KaminskiKM19,FengCSKKZ23,BatzKKM21,McIverMKK18,McIverMorganBook} that uses the $\pGCL$ syntax.
Our choice of $\CFG$s follows the same choice made by \citet{ChatterjeeGMZ22} to describe quantitative termination.
It is standard to compile $\pGCL$ programs into $\CFG$s and vice versa.
For readability, we employ the syntax of $\pGCL$ in some of our examples.
% It is perfectly possible to translate our rules to reason about $\pGCL$ programs.
\end{remark}

\subsubsection*{States, Runs, and Reachable States.}
Fix a $\iCFG$ $\cG = (L, \allowbreak V, \allowbreak l_{init}, \allowbreak x_{init}, \allowbreak \ordMapsTo, \allowbreak G, \allowbreak \sfPr, \allowbreak \sfUp)$.
A \emph{state} is a tuple $(l, \bfx)$, where $l \in L$ and $\bfx \in \setOfRationals^{V}$.
A state $(l, \bfx)$ is termed \emph{assignment} (resp., \emph{nondeterministic}, or \emph{probabilistic}) 
if the location $l$ is assignment (resp., nondeterministic, or probabilistic).
We will refer to assignment locations $l$ where the updates of all transitions sourced at $l$ don't change the variable values as \emph{deterministic} locations.
Accordingly, states $(l, \bfx)$ are termed \emph{deterministic} when $l$ is deterministic.
A transition $(l, l') \in \ordMapsTo$ is \emph{enabled} at a state $(l, \bfx)$ if the guard $G(l, l')$ evaluates to true under $(l, \bfx)$.
% A state is \emph{terminal} if no transitions are enabled at the state.
The vector $\bfx'$ is the \emph{result} of the update pair $(j, u)$ from the state $(l, \bfx)$ if
\begin{inlinelist}
    \item for all $i \neq j$, $\bfx'[i] = \bfx[i]$, and
    \item $\bfx'[j] = u(\bfx)$.
\end{inlinelist}
A state $(l', \bfx')$ is a \emph{successor} to $(l, \bfx)$ if the transition $(l, l') \in \ordMapsTo$ is enabled at $(l, \bfx)$ and $\bfx'$ is the result of $\sfUp(l, l')$ on $\bfx$.
A \emph{finite path} is a sequence of states $(l_1, \bfx_1), (l_2, \bfx_2), \ldots, (l_n, \bfx_n)$ with $(l_{k+1}, \bfx_{k+1})$ being a successor to $(l_k, \bfx_k)$.
A \emph{run} (or \emph{execution}) of $\cG$ is a sequence of states that
\begin{inlinelist}
    \item begins with the initial state $(l_{init}, \bfx_{init})$, and
    \item only induces finite paths as prefixes.
\end{inlinelist}

A state $(l', \bfx')$ is said to be \emph{reachable} from a state $(l, \bfx)$ if there exists a finite path beginning at $(l, \bfx)$ and ending at $(l', \bfx')$.
We write $\reach(\cG, (l, \bfx))$ for the set of states reachable from $(l,\bfx)$; we simply write $\reach(\cG)$ when the initial state is $(l_{\init}, \bfx_{\init})$.
An $\iCFG$ is said to be \emph{finite state} if the set of states reachable from its initial state is finite.

\subsubsection*{Probability Theory}
Let $X$ be any nonempty set.
The tuple $(X, \cF, \Prob)$ is called a \emph{probability space} when $\cF$ is a $\sigma$-algebra over $X$ and $\Prob$ is a probability measure on $\cF$.
A sequence $(\cF_n)$, where $n$ ranges over $\setOfNaturals$, is said to be a \emph{filtration} of the probability space $(X, \cF, \Prob)$ when each $\cF_n$ is a sub $\sigma$-algebra of $\cF$ and $\cF_i \subseteq \cF_j$ for all $i \leq j$.

A \emph{random variable} is a measurable function from $X$ to $\setOfPositiveReals$, the set of positive real numbers.
The \emph{expected value} of the random variable $X$ is denoted by $\bbE[X]$.
A \emph{stochastic process} over the filtered probability space is a sequence of random variables $(X_n)$ such that each $X_n$ is measurable over $\cF_n$.
A \emph{supermartingale} is a stochastic process $(X_n)$ that does not increase in expectation, i.e., $\bbE[X_{n+1}] \leq \bbE[X_n]$ for each $n \in \setOfNaturals$.
We will refer to any function over the state space of a $\CFG$ that doesn't increase in expectation at each execution step as a \emph{supermartingale function}.

We shall make use of the following version of Doob's Martingale Convergence Theorem.
\begin{theorem}[\cite{doobBook}]
    If a supermartingale $(X_n)$ is bounded below, then there almost-surely exists a random variable $X_\infty$ such that
    $
    \Prob\left( X_\infty = \lim_{n\to\infty} X_n \right) = 1$ and $\bbE[X_\infty] \leq \bbE[X_0]
    $.
\end{theorem}

\subsubsection*{Schedulers and Probabilistic Semantics.}
We employ an operational semantics to interpret our programs.
This is standard for $\CFG$s \cite{BaierKatoenBook,ChatterjeeGMZ22}.
The semantics of a $\CFG$ $\cG$ is inferred from a probability space over the runs of $\cG$.
Formally, let $\runs_{\cG}$ be the collection of all executions of $\cG$.
For a finite path $\pi$, let $\cyl_{\cG}(\pi)$ denote the \emph{cylinder set} containing all runs $\rho \in \runs_{\cG}$ such that $\pi$ is a prefix of $\rho$.
Now, call $\cF_\cG$ the smallest $\sigma$-algebra on $\runs_\cG$ containing all cylinder sets of finite paths of $\cG$.

A \emph{scheduler} is a mapping from finite paths ending at nondeterministic states to successors from these states.
We do not impose any measurability or computability conditions on schedulers.
A scheduler $\sched$ induces a probability space over $\runs_{\cG}$ for every $\CFG$ $\cG$.
A finite path (or run) $\pi$ is said to be \emph{consistent} with $\sched$ if for every prefix $\pi'$ of $\pi$ ending at a nondeterministic state, the finite path (or run) obtained by appending the successor state $\sched(\pi')$ to $\pi'$ is a prefix of $\pi$.
A scheduler is said to induce a finite path (or run) if the path (or run) is consistent with the scheduler.
The semantics of the $\iCFG$ $\cG$ under the scheduler $\sched$ is captured by the probability space $(\runs_{\cG}, \allowbreak \cF_\cG, \Prob_\sched)$, where for every consistent finite path $\pi = ((l_1, \bfx_1), (l_2, \bfx_2), \ldots (l_n, \bfx_n))$ with probabilistic locations at indices $i_1, i_2, \ldots i_n$,
$$
\Prob{}_\sched(\pi) = \sfPr(l_{i_1}, l_{i_1 + 1})[\bfx_{i_1}] \times \cdots \sfPr(l_{i_n}, l_{i_n + 1})[\bfx_{i_n}]
$$
We analogously define a probability space $(\runs_{\cG(l, \bfx)}, \allowbreak \cF_{\cG(l, \bfx)}, \allowbreak \Prob_\sched)$ to refer to the probability space induced by the scheduler $\sched$ on the $\iCFG$ obtained from $\cG$ by setting the initial state to $(l, \bfx)$.

For a scheduler $\sched$, the \emph{canonical filtration} of $\cG$ is the sequence $(\cF_n)_{n \in \setOfNaturals}$ such that $\cF_n$ is the smallest sub-$\sigma$-algebra of $\cF_{\cG(\sigma)}$ that contains the cylinder sets $\cyl_{\cG(\sigma)}(\pi_{\leq n})$ of all finite paths $\pi_{\leq n}$ of length at most $n$.
Under this filtration, the semantics of $\cG$ under $\sched$ can also be viewed as a stochastic process $(X^\sched_n)_{n \in \setOfNaturals}$ measurable against $(\cF_n)_{n \in \setOfNaturals}$ such that $X^\sched_n$ takes on an encoding of the state of the execution after $n$ steps.
When convenient, we will use this view as well.

\subsection{The Termination Problem} 

Fix an $\iCFG$ $\cG$ and a scheduler $\sched$. 
Let $l_{\out}$ be a distinguished location that we will call the \emph{terminal} location.
Denote by $\Diamond (l_{\out}, \bfzero)$ the set of all runs of $\cG$ that reach $(l_{\out}, \bfzero)$; we call these the \emph{terminating runs}.
% Denote by $term_{\cG}$ the set of all terminating executions of $\cG$. 
% Observe that $term_{\cG}$ is trivially in $\cF_\cG$. 
Observe that $\Diamond (l_{\out}, \bfzero)$ is measurable.
The $\iCFG$ $\cG$ is said to \emph{terminate} with probability $p$ under the scheduler $\sched$ if $\Prob_\sched[\Diamond (l_{\out}, \bfzero)] = p$. 

\begin{definition}[Termination Probability]
Let $\cG$ be an $\iCFG$ and 
for a scheduler $\sched$, let $(\runs_{\cG}, \cF_\cG, \Prob_\sched)$ be the probability space induced by $\sched$ on the executions of $\cG$. 
The termination probability of $\cG$, denoted by $\termProb(\cG)$, is the infimum of $\Prob_\sched[\Diamond (l_{\out}, \bfzero)]$ over all schedulers $\sched$.
\end{definition}
We use $\termProb(\cG(\sigma))$ to refer to the termination probability of $\CFG$ obtained by setting the initial state of $\cG$ to $\sigma$.
Note that, while we define termination for a specific state $(l_{\out}, \bfzero)$, more general termination conditions can be reduced to this case by a syntactic modification.

An $\iCFG$ is said to be \emph{almost surely terminating} ($\AST$) if its termination probability is $1$.
Our work is on sound and complete proof rules for deciding, for an $\iCFG$ $\cG$,
\begin{inlinelist}
    \item the $\AST$ problem, i.e., whether $\cG$ is almost surely terminating.
    \item the \emph{Lower Bound} problem, i.e., whether $\termProb(\cG)$ exceeds some $p < 1$,
    \item the \emph{Upper Bound} problem, i.e., whether $\termProb(\cG)$ is bounded above by some $p > 0$.
\end{inlinelist}
We remark that, for the lower and upper bound problems, the proof rules we describe are applicable to any number $p$ that is representable in our program logic.
This means that $p$ can take on irrational and transcendental values; this is important, as termination probabilities can often take on such values \cite{BatzKKM21,FlajoletPS11}.
We will elaborate in \cref{subsec:program-logic}.
% Notice that $\AST$ is a special case of the Lower Bound problem; as such, our sound and complete rules for the Lower Bound problem can be used to show $\AST$.
% However, since our Lower Bound rule uses an $\AST$ rule to solve a subproblem, we justify a separate treatment for $\AST$.

\begin{example}[Symmetric Random Walk]
\label{ex:srw}
A $d$-dimensional \emph{symmetric random walk} has $d$ integer variables $x_1, \ldots, x_d$. 
Initially, all variables are $1$.
In each step, the program updates the variables to move to a ``nearest neighbor'' in the $d$-dimensional lattice $\setOfIntegers^d$;
that is, the program picks uniformly at random one of the variables and an element in $\{-1, +1\}$,
and adds the element to the chosen variable.
\codeRef{code:2dsrw} shows the code for $d=2$.
It is well known \cite{Polya} that the symmetric random walk is recurrent in dimension $1$ and $2$, and transient otherwise.
Thus, if we set any element in the lattice, say $0$, to be an terminal state, then the program is almost surely terminating in dimension $1$ and $2$
but not almost surely terminating when $d\geq 3$. 
We shall refer to the $d=1$ and $d=2$ cases as 1DRW and 2DRW, respectively.
\qed

\end{example}
\begin{figure}[t]
    \declareCodeFigure
    \small
    \begin{lstlisting}[language=python, mathescape=true, escapechar=|, xleftmargin=15pt]
x, y $\coloneqq$ 1, 1
while (x |$\neq$| 0 |$\lor$| y |$\neq$| 0):
    { x $\coloneqq$ x + 1 |$\pChoice{\frac{1}{2}}$| x $\coloneqq$ x - 1 } |$\pChoice{\frac{1}{2}}$|
        { y $\coloneqq$ y + 1  |$\pChoice{\frac{1}{2}}$| y $\coloneqq$ y - 1 }
    \end{lstlisting}
    \caption{The 2D symmetric random walker. \textnormal{The symbol $\pChoiceWithoutParamters$ is a probabilistic choice operator.}}
    \label{code:2dsrw}
    \finishcodefigure
\end{figure}

\subsection{The Unrolling Lemma}
\label{subsec:unrolling-lemma}

Let $\cG$ be an $\iCFG$ such that $\termProb(\cG) \geq p$ for some rational $p > 0$.
Fix a scheduler $\sched$.
Let $(\pi_1, \pi_2, \ldots)$ be an ordering of the terminating runs of $\cG$ consistent with $\sched$ such that $|\pi_1| \leq |\pi_2| \leq \cdots$.
For some $\epsilon > 0$, let $i_n$ be the smallest number such that
$$
\Prob_\sched(\pi_1) + \cdots + \Prob_\sched(\pi_{i_n}) \geq p - \epsilon
$$
where $\Prob_\sched$ is the probability measure induced by $\sched$ over the set of all runs of $\cG$.

We call $|\pi_n|$ the \emph{required simulation time} of $\cG$ under $\sched$ to assimilate a termination probability of $p - \epsilon$.
The required simulation time of $\sched$ is simply the length of the longest terminating run that must be accounted for in the termination probability series for it to cross $p - \epsilon$.
Define the simulation time of $\cG$ w.r.t. $\epsilon$ as the supremum over all schedulers $\sched$ of the required simulation time of $\cG$ under $\sched$ and $\epsilon$.
The following lemma is at the core of showing that the almost sure termination problem is $\Pi^0_2$-complete \cite{MS24arxiv}.

\begin{lemma}[Unrolling Lemma \cite{MS24arxiv}]
    \label{lem:required-simulation-time}
    Let $\cG$ be an $\iCFG$ such that $\termProb(\cG) \geq p$.
    For any $\epsilon$, the simulation time of $\cG$ w.r.t.\ $\epsilon$ is bounded above.
\end{lemma}

\cref{lem:required-simulation-time} is a generalization of Lemma B.3 of \citet{MS24arxiv}.
It holds for $\iCFG$s because the branching at nondeterministic locations is bounded.
To prove the unrolling lemma, for each $m\in\nat$, we consider unrollings of $\cG$ for $m$ steps, running under partial schedules that resolve nondeterministic choices for up to $m$ steps.
Partial schedules are naturally ordered into a tree, where a partial schedule $\sched$ is extended by $\sched'$ if $\sched'$ agrees with $\sched$ when restricted to the domain of $\sched$.
An infinite path in this tree defines a scheduler.
For each scheduler $\sched$, we mark the $k$-th node in its path if $k$ is the minimum number such that the $k$-step unrolled program amasses termination probability 
at least $p - \epsilon$.
If two schedulers agree up to $k$ steps, then they both mark the same node.
The key observation is that, since the nondeterminism is finite-branching, the scheduler tree is finite-branching.
Thus, if we cut off the tree at marked nodes and still have an infinite number of incomparable marked nodes, there must be an infinite path in the tree that is not marked.
But this is a contradiction, because this infinite path corresponds to a scheduler that never amasses $p - \epsilon$ probability mass for termination.

% \cref{cor:bound-on-shortest-term-runs} follows directly from \cref{lem:required-simulation-time}.%, and is used multiple times in the proofs of the soundness and completeness of our rules.

\begin{corollary}
    \label{cor:bound-on-shortest-term-runs}
    Let $\sigma$ be a state of an $\iCFG$ $\cG$.
    Suppose $\termProb(\cG(\sigma \allowbreak))> 0$.
    Then, varied across schedulers, there is an upper bound on the length of the shortest consistent terminal run from $\sigma$.
\end{corollary}
\begin{proof}
    For a scheduler $\sched$, let $\pi_\sched$ be the smallest terminal run of $\cG(\sigma)$ consistent with $\sched$.
    The simulation time to assimilate a termination probability of $\epsilon$ for some $0 < \epsilon < \termProb(\cG(\sigma))$ under scheduler $\sched$ is necessarily at least as large as $|\pi_\sched|$.
    If the collection of lengths $|\pi_\sched|$ across schedulers $\sched$ wasn't bounded, then this simulation time is unbounded.
    This contradicts the unrolling lemma.
\end{proof}

\subsection{Assertion Language and Program Logic}
\label{subsec:program-logic}

Our language of choice for specifying assertions is the first-order language of arithmetic with addition, multiplication, and order interpreted over the domain of rationals.
% This choice is standard in program logics \cite{Apt81}.
We fix the interpretation model for our assertions as the standard model of rationals.
Refer to this interpretation by $\intQ$. \footnote{
    Instead of fixing $\intQ$, one can use any \emph{arithmetical structure} \cite{HarelKozenTiuryn} to specify and interpret assertions.
    All our proof rules will remain sound and relatively complete with this change.
}
Let $\ThQ$ denote the \emph{theory of rationals}, i.e., the collection of assertions that are true in the standard model of rationals.
The evaluation of our assertions is tantamount to their implication by $\ThQ$.
All proof techniques we present in our work are relative to complete proof systems for $\ThQ$.

Assertions are evaluated at program states.
Fix a $\iCFG$ $\cG$ with transition relation $\ordMapsTo_\cG$.
A state $\sigma$ of $\cG$ \emph{satisfies} an assertion $\varphi$ if the interpretation $\intQ$ augmented with the variable valuation encoded in $\sigma$ models $\varphi$.
We denote this by $\sigma \vDash \varphi$.
% Note that, because we fixed the interpretation to $\intQ$, we do not mention the interpretation here. 
An assertion $\varphi$ is \emph{valid} if $\sigma \vDash \varphi$ for all states $\sigma$.
Valid assertions are contained in $\ThQ$.

We employ a program logic inspired by the seminal work of \citet{Floyd1993}.
Statements in our logic affix assertions as preconditions and postconditions to transitions in $\cG$.
For example, the transition $\tau \in \ordMapsTo_\cG$ could be affixed a precondition $\varphi_\tau$ and postcondition $\psi_\tau$ to yield the sentence $\{\varphi_\tau\} \tau \{\psi_\tau\}$.
The precondition $\varphi_\tau$ is evaluated at the program state before taking $\tau$, and the postcondition $\psi_\tau$ is evaluated at states reached immediately after $\tau$.
The sentence $\{\varphi_\tau\} \tau \{\psi_\tau\}$ is \emph{true} for $\cG$ if for every state $\sigma$ with $\sigma \vDash \varphi$, if $\tau$ is enabled at $\sigma$ and $\sigma'$ is a successor of $\sigma$ through $\tau$, then $\sigma' \vDash \psi_\tau$.

We use the notion of \emph{inductive invariants} in our proof rules.
An \emph{inductive invariant} is an assertion with $n+1$ free variables, the first ranging over $L$ and the others over $\setOfRationals$, that is closed under the successor operation.
That is, an assertion $\Inv$ is an inductive invariant if, whenever $(l, \bfx)$ satisfies $\Inv$, and $(l', \bfx')$ is a successor to $(l, \bfx)$, then $(l', \bfx')$ satisfies $\Inv$.
It follows that if $(l_{\init}, \bfx_{\init})$ satisfies $\Inv$, then every reachable state satisfies $\Inv$.
In this paper, we will mildly abuse notation and use $\Inv$ to also refer to the set of states satisfying $\Inv$.
% Moreover, just as in standard Hoare logic, given an assertion $\Inv$, we can reduce the check that it is an inductive invariant to a validity question in the underlying logic.

\citet{Floyd1993} specified axioms for a proof system over this program logic.
\emph{Proof rules} extend this system by enabling the deduction of nontrivial program properties, such as termination.
These rules are composed of \emph{antecedents} and \emph{consequents}.
Antecedents are finite collections of statements written in the program logic.
Consequents detail properties of the program over which the antecedents are evaluated.
\emph{Soundness} of a proof rule means that if the antecedents are true for a program $\cG$, then the consequents hold for $\cG$.
\emph{Completeness} of a proof rule means that if the consequents are true for some program $\cG$, then one can come up with proofs for the validity of the antecedents of the rule over $\cG$ in the underlying proof system.

All proof rules specified in this paper operate over a $\CFG$ denoted by
$\cG = (L, \allowbreak V, \allowbreak l_{init}, \allowbreak \bfx_{init}, \allowbreak \ordMapsTo, \allowbreak G, \allowbreak \sfPr, \allowbreak \sfUp)$.
Their completeness is dependent on the existence of a complete proof system for $\ThQ$.
Such proof rules are said to be complete \emph{relative to a proof system for $\ThQ$}.
Relative completeness of this kind is standard in program logics.

To demonstrate the relative completeness of our proof rules, we will need to encode \emph{computable relations} in our assertion language.
A relation is computable if its characteristic function is decidable.
It is classically known that the theory of arithmetic interpreted over natural numbers can encode any computable relation.
Let $\intN$ refer to the interpretation model of the standard model of naturals.
Denote by $\ThN$ the collection of all true assertions under $\intN$.
$\ThN$ is generally referred to as the theory of natural numbers.
For each computable relation $R(x_1, x_2, \ldots x_n)$, there is an assertion $\varphi_R(x_1, x_2, \ldots x_n)$ that is true in $\ThN$.
To represent computable relations in $\ThQ$, we use a result by \citet{Robinson49}.
\begin{theorem}[\citet{Robinson49}]
    \label{theorem:robinson}
    $\setOfNaturals$ is definable in $\ThQ$.
\end{theorem}
We refer to the assertion that encodes $\setOfNaturals$ by $\Nat$.
Therefore, $\Nat(x)$ is true in $\intQ$ \emph{iff} $x \in \setOfNaturals$.
All computable relations can be encoded in our assertion language through liberal usage of $\Nat$.
An important implication is that termination probabilities are expressible in our assertion language.
\begin{lemma}
    For a $\CFG$ $\cG$ and a $p \in [0, 1]$ with $\termProb(\cG) = p$, there is an assertion $\psi(x)$ with one free variable $x$ such that $\ThQ \vDash \psi(x) \Leftrightarrow x \leq p$.
\end{lemma}
\begin{proof}
    We know that $\termProb(\cG) \geq p$ \emph{iff} $\termProb(\cG) \geq p - \epsilon$ for all $\epsilon > 0$.
    The unrolling lemma implies that for all $\epsilon > 0$, there is a $k \in \setOfNaturals$ such that the probability mass of the $k$-unrolled program is at least $p - \epsilon$.
    Finite unrollings of $\cG$ are, by definition, computable, and checking the probability of termination amassed in this finite unrolling is also computable; see \citet{KaminskiKM19} for details.
    This means that a relation $R(\epsilon, k)$ representing this relationship between every rational $\epsilon$ and natural $k$.
    Such computable relations are representable in $\ThQ$ through \cref{theorem:robinson}. 
\end{proof}
Notice that while the termination probabilities $p$ are real numbers, the lower bounds verified by the assertion $\psi$ in the above lemma are entirely rational.
However, by representing the set of rational numbers under $p$, $\psi$ has effectively captured the Dedekind cut of $p$.
This expressibility shows how our proof rules can use irrational lower bounds on termination probabilities.

%% file: ast-martingale.tex
\section{Almost-Sure Termination}
\label{sec:ast}

Our proof rules consist of sets and functions over the state spaces that satisfy certain properties.
Each of these entities must be representable in our assertion language; therefore, they must be arithmetical expressions over the program variables and program locations.
Instead of specifying each condition in our rules as formal statements in our program logic, we directly describe the properties these entities must satisfy.
We do so to emphasize these entities themselves over the formalism surrounding them.
It is nevertheless possible to write each of the following proof rules as finite sets of statements in our program logic.

Recall that a proof rule is \emph{sound} if, whenever we can find arithmetical expressions in our assertion language that satisfy the conditions outlined in the premise of a rule, the conclusion of the rule holds.
A proof rule is \emph{relatively complete} if, whenever the conclusion holds (e.g., a program $\cG$ is $\AST$),
we can find certificates in the assertion language that satisfy all the premises.

% The notion of the Inductive Invariant also plays an important role in our rules.
% Invariants are simply collections of states canonically used to overapproximate the set of reachable states.
% They are common tools in program analysis \cite{McIverMorganBook, Dijkstra76}.

\subsection{McIver and Morgan's Variant Rule}
\label{subsec:mciver-morgan-rule}

We start with a well-known rule for almost-sure termination by \citet{McIverMorganBook}.
The rule is sound but complete only for finite-state programs \cite[Lemma 7.6.1]{McIverMorganBook}.

\begin{proofrule}[Variant Rule for $\AST$ \cite{McIverMorganBook}] 
    \label{proofrule:ast-variant}
    To show that $\cG$ is $\AST$, find
    \begin{enumerate}
        \item an inductive invariant $\Inv$ containing the initial state $(l_{init}, \bfx_{init})$,
        \item a \emph{variant} function $U : \Inv \to \setOfIntegers$,
        \item bounds $\Lo$ and $\Hi$ such that for all states $(l, \bfx) \in \Inv$, $\Lo \leq U(l, \bfx) < \Hi$, and
        \item an $\epsilon > 0$,
    \end{enumerate}
    such that, for each state $(l, \bfx) \in \Inv$,
    \begin{enumerate}[label=(\alph*)]
        \item if $(l, \bfx)$ is a terminal state, $U(l, \bfx) = \Lo$.
        \item if $(l, \bfx)$ is an assignment, or nondeterministic state, $U(l', \bfx') < U(l, \bfx)$ for every successor $(l', \bfx')$.
        \item if $(l, \bfx)$ is a probabilistic state, $\sum \sfPr(l, l')[\bfx] > \epsilon$ over all successor states $(l', \bfx')$ with $U(l', \bfx') < U(l, \bfx)$.
    \end{enumerate}
\end{proofrule}
The function $U$ asked by the rule maps states to integers and decreases with probability at least $\epsilon$ at each execution step.
Intuitively, the value it assigns to a state corresponds to the length of the shortest terminal run from that state.
Notice that this length must necessarily decrease by $1$ at one direction in each transition.
$U$ thus measures the distance a state is from termination, and is generally referred to as a \emph{distance variant}.
Condition (3) then implies that the distance to termination across states in $\cG$ is bounded above.
\begin{lemma}[\citet{McIverMorganBook}]
    \label{lem:completeness-mciver-morgan-rule}
    \cref{proofrule:ast-variant} is sound for all $\AST$ programs.
    It is relatively complete for finite-state $\AST$ $\CFG$s.
\end{lemma}
Soundness of the rule follows from one the application of a zero-one law of probability theory to the fact that the distance to termination across states is bounded above by $\Hi - \Lo$.
Finite-state completeness results from the necessary upper bound (equal to the number of states) on the length of shortest terminal runs beginning from each state in $\AST$ programs.
While \citet{McIverMorganBook} claim completeness and not relative completeness, their proof trivially induces relative completeness.

This rule is not complete, however.
This is because, if the rule is applicable, the program is guaranteed a terminal run of length at most $\Hi - \Lo$ from any state.
But the 1D random walk (outlined in \cref{ex:srw}) does not satisfy this property, even though it terminates almost-surely.

% Over the years, McIver and Morgan's proof rule has been extended many times \cite{McIverMKK18,HuangFC18}.
% The most significant extension is the proof rule of \citet{McIverMKK18}, in which they require the function $U$ to additionally be a supermartingale.
% We will discuss this extension in \cref{subsec:mciver-katoen-rule}.

% These modifications maintain the underlying supermartingale
% condition that there is a function $U$ on the state space that
% decreases in expectation by some amount $\epsilon > 0$ in each step.
% These extensions effectively retain the function $U$ on the state space that is guaranteed a non-zero lower bound over the probability of its decrease in each step.
% In particular, at assignment states, $U$ must certainly decrease.
% The following example shows that these rules cannot be complete.
% The essential idea of the example goes back to
% \cite{Back81,AptP86} to show that integer-valued variants
% are not sufficient for proving program termination in the presence of
% unbounded nondeterminism.
% We also note that a more complex example appeared in \cite{MS24} to
% show ranking supermartingales, which expect to reduce by some $\epsilon > 0$ in each step, are incomplete for positive almost-sure
% termination,
% where we require the expected time to termination be finite for every scheduler.

\subsection{Our Rule}
\label{sec:ast-martingale}

We present here a martingale-based proof rule for $\AST$ that exploits the fact that $\AST$ programs, when repeatedly run, are recurrent.
\begin{proofrule}[Martingale Rule for $\AST$]
    \label{proofrule:ast-martingale}
    To prove that $\cG$ is $\AST$, find
    \begin{enumerate}
    \item an inductive invariant $\Inv$ containing the initial state,
    % \item a set $A \subset \Inv$ containing the terminal state, % $A$ can be infinite now, but it must strictly be less than Inv
    \item a supermartingale function $V : \Inv \to \setOfReals$ that assigns $0$ to the terminal state and at all states $(l, \bfx) \in \Inv$, %\setminus A$,
        \begin{enumerate}
            \item $V(l, \bfx) > 0$,
            % \item $V(l, \bfx) > V(l_A, \bfx_A)$ for each $(l_A, \bfx_A) \in A$,
            \item if $(l, \bfx)$ is an assignment or nondeterministic state, then $V(l, \bfx) \geq V(l', \bfx')$ for all possible successor states $(l', \bfx')$, and
            \item if $(l, \bfx)$ is a probabilistic state, then, over all successor states $(l', \bfx')$, $V(\sigma) \geq \sum \Pr(l, l') \allowbreak V(l', \bfx')$,
        \end{enumerate}
    \item a variant function $U : \Inv \to \setOfNaturals$ that assigns $0$ to the terminal state,
        \begin{enumerate}
            % \item assigns $0$ to the terminal state,
            \item ensures that at nondeterministic and assignment states $(l, \bfx) \in \Inv$,  $U(l, \bfx) > U(l', \bfx')$ for all possible successor states $(l', \bfx')$, and
            \item satisfies the following compatibility criteria with the sublevel sets $V_{\leq r} = \{ \sigma \in \Inv \mid V(\sigma) \leq r \}$ for each $r \in \setOfReals$:
            \begin{enumerate}
                \item the set $\{ u \in \setOfNaturals \mid \sigma \in V_{\leq r} \land u = U(\sigma) \}$ is bounded, and
                \item there exists an $\epsilon_r > 0$ such that, for all probabilistic states $(l, \bfx) \in V_{\leq r}$, the sum $\sum \sfPr(l, l')[\bfx] > \epsilon_r$ over all successor states $(l', \bfx')$ with $U(l', \bfx') < U(l, \bfx)$.
            \end{enumerate}
        \end{enumerate}
    \end{enumerate}
% 
%     Under these conditions, $\cG$ is $\AST$.
\end{proofrule}
In this rule, $U$ is meant to play the role of the variant function from \cref{proofrule:ast-variant}.
% $A$ is meant to form a ``ball'' around the terminal state; it is useful in applications where the supermartingale properties of $V$ are difficult to establish at all states.
% If the execution were to be restricted within this ball $A$, the rule makes it easy to establish almost-sure termination.
% This is because while $V$ must only be a supermartingale outside of $A$, the variant $U$ must still decrease within $A$.
% Observe that $A$ must be a strict subset of $\Inv$; this is to enforce an upper bound on the collection of $V$-values of states in $A$.
Accordingly, $U$ assigns to states the length of shortest terminal run begnning at that state.
The supermartingale $V$, on the other hand, can be thought of as a measure of \emph{relative likelihood}.
The probability that a transition increases $V$ by an amount $v$ reduces as $v$ increases.
Unlikely transitions are associated with greater increments to $V$, and (relatively) unlikely states have greater $V$ values.
Notice that this rule reduces to the prior \cref{proofrule:ast-variant} if the supermartingale $V$ was bounded.
% Separately, the variant $U$ reprises its role from \cref{proofrule:ast-variant}: it effectively measures the shortest distance to a terminal state.

At a high level, this rule works for the following reasons.
Suppose $V$ is unbounded and executions begin at some initial state $\sigma_0 \in \Inv$.
The supermartingale property of $V$ implies that from $\sigma_0$, the probability of reaching a state $\sigma$ with $V(\sigma) > V(\sigma_0)$ approaches $0$ as $V(\sigma)$ grows to $+\infty$.
Now, fix an unlikely state $\sigma$ with $V(\sigma) \gg V(\sigma_0)$.
Let's now restrict our attention to the executions that remain in states $\gamma \in \Inv$ with $V(\gamma) \leq V(\sigma)$.
The compatibility conditions satisfied by the variant $U$ with $V$ at the sublevel set $V_{\leq V(\sigma)}$ implies the almost-sure termination of these executions.
The remaining executions must reach some unlikely state $\gamma'$ with $V(\gamma') \geq V(\sigma)$.

Thus, as the probability of reaching unlikely states $\gamma'$ reduces the ``further away'' (from the perspective of $V$) they are, the probability of terminating approaches $1$.
Since $V$ is unbounded, the probability of termination is $1$.

% \emph{Remark.}
% We note that our rule is quite similar to the $\AST$ rule of \citet{McIverMKK18}.
% Their rule consisted of a single supermartingale $V$ that, with the help of a few antitone functions, also exhibited the properties of a distance variant.
% In other words, they combined the duties of the functions $V$ and $U$ into a single function $V$.
% \citet{McIverMKK18} were unable to prove the completeness of their rule.
% We show in the later \cref{subsec:mciver-katoen-rule} that our techniques imply the completeness of their rule.

\begin{lemma}[Soundness]
    \cref{proofrule:ast-martingale} is sound. 
\end{lemma}

\begin{proof}
    Let us first dispense of the case where $V$ is bounded.
    If $V$ is bounded, the compatibility criteria forces a bound on the variant function $U$.
    The soundness of \cref{proofrule:ast-variant} implies $\cG \in \AST$.
    Therefore, from now on, $V$ is assumed to be unbounded.

    Denote the initial state by $\sigma_0$.
    For each $n \in \setOfNaturals$, define $\Pi_n$ to be the collection of runs from $\sigma_0$ that reach a maximum $V$ value of $n$.
    This means that for each state $\sigma$ encountered in executions in $\Pi_n$, $V(\sigma) \leq n$.
    Define $\Pi_\infty$ to be the remaining collection of executions beginning at $\sigma_0$ that don't have a bound on the $V$ values that they reach.
    This means that for each execution $\pi \in \Pi_\infty$ and each $n \in \setOfNaturals$, there are states $\sigma \in \pi$ such that $V(\sigma) > n$.
    We have thus partitioned the collection of executions of the $\iCFG$ $\cG$ to $\Pi_\infty \cup (\bigcup_{i \in \setOfNaturals} \Pi_i)$.

    We will now argue that under every scheduler, the probability measure of all non-terminating executions in each $\Pi_n$ is $0$.
    By definition, all executions in $\Pi_n$ lie entirely within the sublevel set $V_{\leq n}$.
    The compatibility of $U$ with $V_{\leq n}$ implies that the variant $U$ is bounded across states in $\Pi_n$.
    Consider an $\iCFG$ $\cG_{\leq n}$ that mirrors $\cG$ inside $V_{\leq n}$, but marks states in $\cG$ outside $V_{\leq n}$ as terminal.
    Applying \cref{proofrule:ast-variant} using the now bounded variant $U$ allows us to deduce that $\cG_{\leq n}$ is almost-surely terminating.
    Observe now that the collection of non-terminating runs of $\cG_{\leq n}$ is precisely the collection of non-terminating runs in $\Pi_n$.
    This immediately gives us what we need.

    We now turn our attention to the final collection $\Pi_\infty$.
    Observe that $\Pi_\infty$ must only contain non-terminal executions.
    Suppose that, under some scheduler $\sched$, the probability measure of $\Pi_\infty$ wasn't $0$.
    Let the probability space defining the semantics of $\cG$ under $\sched$ (see \cref{subsec:syntax-semantics}) be $(\runs_{\cG(\sigma)}, \cF_{\cG(\sigma)}, \Prob_\sched)$, and let its canonical filtration be $\{\cF_n\}$.
    Define a stochastic process $\{ X^\sched_n \}$ over the aforementioned probability space augmented with the filtration $\{ \cF_n \}$ that tracks the current state of the execution of the program.
    Define another stochastic process $\{ Y^\sched_n \}$ as $Y^\sched_n \triangleq V(X^\sched_n)$ for each $n \in \setOfNaturals$.
    % $$
    % Y^\sched_n \triangleq \begin{cases}
    %     V(X^\sched_n) & X^\sched_n \text{ is not terminal} \\
    %     0 & \text{otherwise}
    % \end{cases}
    % $$
    It's easy to see that $Y^\sched_n$ is a non-negative supermartingale.
    Since $Y^\sched_n$ is non-negative, Doob's Martingale Convergence Theorem \cite{doobBook} implies the almost-sure existence of a random variable $Y^\sched_\infty$ that the process $\{Y^\sched_n\}$ converges to.
    This means that $\bbE[Y^\sched_\infty] \leq Y^\sched_0$.

    Under the condition that $\Pi_\infty$ occurs, $Y^\sched_\infty = +\infty$.
    Since the probability measure of these non-terminal executions isn't $0$, we have that $\bbE[Y^\sched_\infty] = +\infty > Y^\sched_0 = V(X^\sched_0)$.
    This raises a contradiction, completing the proof.
\end{proof}

To show completeness, we adapt a technique by \citet{MSZ78} to build the requisite supermartingale $V$.
Suppose $\cG$ is $\AST$.
Let $\reach(\cG)$ be the set of its reachable states.
Fix a computable enumeration $\sfEnum$ of $\reach(\cG)$ that assigns $0$ to its terminal state.
Intuitively, $\sfEnum$ is meant to order states in a line so that the probability of reaching a state that's far to the right in this order is small.
This is because the $\AST$ nature of $\cG$ forces executions to ``lean left'' toward the terminal state.
Note that we place no other requirements on $\sfEnum$; these intuitions will work no matter how $\sfEnum$ orders the states.
A state $\sigma$ is said to be indexed $i$ if $\sfEnum(\sigma) = i$.
From now on, we will refer to the state indexed $i$ by $\sigma_i$.

A crucial part of our construction is the following function $R : (\setOfNaturals \times \setOfNaturals) \to [0, 1]$. %TODO: I'd like to name $R$ properly to place it inside a definition block, but I don't like most of the names I was able to come up. Maybe "Scope", "Reach", or "Potential" could work? It has to be a name that I can use in a sentence like: "the Potential of a state tends to $0$ as the set of target indices tends to infinity."
Intuitively, $R$ measures the ability of executions beginning from a state to reach states that are far to the right of it in the $\sfEnum$ order.
Let $\cG_i$ be the $\CFG$ obtained from $\cG$ by switching its initial state to $\sigma_i$.
Let the semantics of $\cG_i$ under a scheduler $\sched$ be the probability space $(\runs_{\cG_i}, \allowbreak \cF_{\cG_i}, \allowbreak \Prob^i_\sched)$.
Define $R(i, n)$ at indices $i$ and $n$ to be
\begin{equation}
    \label{eq:scope-defn}
    R(i, n) \triangleq \sup{}_\sched \Prob^i_\sched\left(\Diamond\left( \left\{\sigma_m \in \reach(\cG) \mid m \geq n\right\}\right)\right)
\end{equation}
Where $\Diamond(C)$ represents the event of eventually reaching the set $C$.
We will refer to the first argument $i$ as the source index and the second argument $n$ as the minimum target index.
Put simply, $R(i, n)$ is the supremum probability of reaching the target indices $\{n, n+1, \ldots\}$ from the source $\sigma_i$.

\begin{lemma}
    \label{lem:scope-usefulness}
    $R(i, n) \to 0$ as $n \to \infty$ at every $i \in \setOfNaturals$, i.e.,
    $
    \forall i \in \setOfNaturals \cdot \lim_{n\to\infty} R(i, n) = 0
    $.
\end{lemma}
\begin{proof}
    Denote by $E_n$ the event that executions beginning from $\sigma_i$ reach states with index $\geq n$.
    Clearly, $R(i, n)$ measures the supremum of the probability of $E_n$ across all schedulers.
    % Denote by $E_\infty$ the event that executions beginning from $\sigma_i$ increase the maximum observed state index infinitely often.
    % It's easy to see that, for every $n \in \setOfNaturals$, the event $E_n$ contains $E_\infty$.
    % Also, each execution outside $E_\infty$ must be inside some $E_n \setminus E_{n+1}$, as it must yield a maximum state index contained in it.
    % Additionally, $E_{n+1} \subseteq E_n$ for all $n$.
    % These three facts imply
    % $$
    % E_\infty = \bigcap_{i \in \setOfNaturals} E_n = \lim_{n \to \infty} E_n 
    % $$
    Suppose that $\lim_{n \to \infty} R(i, n) > 0$ for some index $i$.
    This strict inequality means there is a small $\epsilon > 0$ such that $\lim_{n \to \infty} R(i, n) > \epsilon$.
    Observe that $E_{n + 1} \subseteq E_n$ for all $n$.
    This means that $R(i, n)$ is non-increasing for increasing $n$, which in turn means that $R(i, n) > \epsilon$ for all $n \in \setOfNaturals$.
    Since $R(i, n) = \sup{}_\sched \Prob^i_\sched [E_n]$, there must exist, for each $n$, a scheduler $\sched_n$ such that $\Prob^i_{\sched_n}[E_n] > \epsilon$.
    
    Since $\cG$ is $\AST$, the unrolling lemma indicates the existence of a $k_\epsilon \in \setOfNaturals$ with the property that for every scheduler $\sched$, the probability mass of terminating executions of length $\leq k_\epsilon$ beginning at $\sigma_i$ consistent with $\sched$ is $> 1 - \epsilon$.
    Because the branching at nondeterministic and probabilistic states in our program model is bounded, only a finite collection of states can be reached by executing $\cG$ for up to $k_\epsilon$ states.
    This means that across all schedulers, there is an index $m_\epsilon$ such that no state of index $\geq m_\epsilon$ was reached within $k_\epsilon$ steps.

    Take the scheduler $\sched_{m_\epsilon}$ such that $\Prob^i_{\sched_{m_\epsilon}} [E_{m_\epsilon}] > \epsilon$.
    This means that the probability mass of terminating executions of $\cG$ consistent with $\sched_{m_\epsilon}$ beginning at $\sigma_i$ that do not reach a state with index $\geq m_\epsilon$ is $< 1 - \epsilon$.
    For $\cG$ to amass a termination probability of $\geq 1 - \epsilon$ under these conditions, it must enter some state with index $\geq m_\epsilon$.
    This is not possible within $k_\epsilon$ steps, raising a contradiction and completing the proof.
\end{proof}
It turns out that, if we fix the minimum target index $n$, the function $R$ becomes a supermartingale.
Define $V_n(\sigma) = R(\sfEnum(\sigma), n)$ for every $n \in \setOfNaturals$.
It's easy enough to see that $V_n$ is a supermartingale:
due to the supremum nature of $R(i, n)$, at assignment / non-deterministic states $\sigma$ with possible successors $\sigma'$, $V_n(\sigma) \geq V_n(\sigma')$, and due to the probabilistic nature of $\Prob^i_\sched$,
at probabilistic $\sigma = (l, \bfx)$, $V_n(l, \bfx) \geq \sum \sfPr(l, l') V_n(l', \bfx)$ across all successors $(l', \bfx)$. 
However, $V_n$ isn't the supermartingale we need, as we may not always be able to construct a compatible $U$ for any $V_n$.
This is because every $V_n$ is bounded above (by $1$), whereas $U$ typically isn't bounded above.

\begin{figure}[t]

\centering

\tikzset{every picture/.style={line width=0.75pt}} %set default line width to 0.75pt        

\begin{tikzpicture}[x=0.75pt,y=0.75pt,yscale=-1,xscale=1]
%uncomment if require: \path (0,250); %set diagram left start at 0, and has height of 250

%Shape: Ellipse [id:dp15930781734771693] 
\draw  [draw=none][fill={rgb, 255:red, 248; green, 231; blue, 28 }  , opacity=0.25 ] (25.89,160) .. controls (25.89,132.61) and (81.41,110.4) .. (149.89,110.4) .. controls (218.37,110.4) and (273.89,132.61) .. (273.89,160) .. controls (273.89,187.39) and (218.37,209.6) .. (149.89,209.6) .. controls (81.41,209.6) and (25.89,187.39) .. (25.89,160) -- cycle ;
%Shape: Ellipse [id:dp632196707073734] 
\draw  [draw=none][fill={rgb, 255:red, 255; green, 0; blue, 0 }  , opacity=0.25 ] (410,160.1) .. controls (410,132.76) and (465.29,110.6) .. (533.5,110.6) .. controls (601.71,110.6) and (657,132.76) .. (657,160.1) .. controls (657,187.44) and (601.71,209.6) .. (533.5,209.6) .. controls (465.29,209.6) and (410,187.44) .. (410,160.1) -- cycle ;
%Shape: Rectangle [id:dp38472122039193035] 
\draw  [draw=none][fill={rgb, 255:red, 255; green, 255; blue, 255 }  ,fill opacity=1 ] (543,66.75) -- (667,66.75) -- (667,215.98) -- (543,215.98) -- cycle ;
%Shape: Circle [id:dp40562460433871816] 
\draw  [color={rgb, 255:red, 73; green, 140; blue, 0 }  ,draw opacity=1 ][line width=1.5]  (40,150.11) .. controls (40,144.65) and (44.43,140.22) .. (49.89,140.22) .. controls (55.35,140.22) and (59.78,144.65) .. (59.78,150.11) .. controls (59.78,155.57) and (55.35,160) .. (49.89,160) .. controls (44.43,160) and (40,155.57) .. (40,150.11) -- cycle ;
%Shape: Circle [id:dp3715873207371224] 
\draw   (90,150.11) .. controls (90,144.65) and (94.43,140.22) .. (99.89,140.22) .. controls (105.35,140.22) and (109.78,144.65) .. (109.78,150.11) .. controls (109.78,155.57) and (105.35,160) .. (99.89,160) .. controls (94.43,160) and (90,155.57) .. (90,150.11) -- cycle ;
%Shape: Circle [id:dp009048408016381182] 
\draw   (140,150.11) .. controls (140,144.65) and (144.43,140.22) .. (149.89,140.22) .. controls (155.35,140.22) and (159.78,144.65) .. (159.78,150.11) .. controls (159.78,155.57) and (155.35,160) .. (149.89,160) .. controls (144.43,160) and (140,155.57) .. (140,150.11) -- cycle ;
%Shape: Circle [id:dp3223002682608248] 
\draw   (241,150.11) .. controls (241,144.65) and (245.43,140.22) .. (250.89,140.22) .. controls (256.35,140.22) and (260.78,144.65) .. (260.78,150.11) .. controls (260.78,155.57) and (256.35,160) .. (250.89,160) .. controls (245.43,160) and (241,155.57) .. (241,150.11) -- cycle ;
%Shape: Circle [id:dp6457297541155385] 
\draw   (421,150.11) .. controls (421,144.65) and (425.43,140.22) .. (430.89,140.22) .. controls (436.35,140.22) and (440.78,144.65) .. (440.78,150.11) .. controls (440.78,155.57) and (436.35,160) .. (430.89,160) .. controls (425.43,160) and (421,155.57) .. (421,150.11) -- cycle ;
%Shape: Circle [id:dp7938747542548564] 
\draw   (290,150.11) .. controls (290,144.65) and (294.43,140.22) .. (299.89,140.22) .. controls (305.35,140.22) and (309.78,144.65) .. (309.78,150.11) .. controls (309.78,155.57) and (305.35,160) .. (299.89,160) .. controls (294.43,160) and (290,155.57) .. (290,150.11) -- cycle ;
%Shape: Circle [id:dp3179166858140825] 
\draw   (470,150.11) .. controls (470,144.65) and (474.43,140.22) .. (479.89,140.22) .. controls (485.35,140.22) and (489.78,144.65) .. (489.78,150.11) .. controls (489.78,155.57) and (485.35,160) .. (479.89,160) .. controls (474.43,160) and (470,155.57) .. (470,150.11) -- cycle ;
%Straight Lines [id:da1280832159869264] 
\draw    (90,150.11) -- (61.78,150.11) ;
\draw [shift={(59.78,150.11)}, rotate = 360] [color={rgb, 255:red, 0; green, 0; blue, 0 }  ][line width=0.75]    (10.93,-3.29) .. controls (6.95,-1.4) and (3.31,-0.3) .. (0,0) .. controls (3.31,0.3) and (6.95,1.4) .. (10.93,3.29)   ;
%Straight Lines [id:da29204800869895997] 
\draw    (109.78,150.11) -- (138,150.11) ;
\draw [shift={(140,150.11)}, rotate = 180] [color={rgb, 255:red, 0; green, 0; blue, 0 }  ][line width=0.75]    (10.93,-3.29) .. controls (6.95,-1.4) and (3.31,-0.3) .. (0,0) .. controls (3.31,0.3) and (6.95,1.4) .. (10.93,3.29)   ;
%Curve Lines [id:da6424811267589389] 
\draw    (149.89,140.22) .. controls (122.28,107.7) and (86.72,109.33) .. (50.97,139.3) ;
\draw [shift={(49.89,140.22)}, rotate = 319.49] [color={rgb, 255:red, 0; green, 0; blue, 0 }  ][line width=0.75]    (10.93,-3.29) .. controls (6.95,-1.4) and (3.31,-0.3) .. (0,0) .. controls (3.31,0.3) and (6.95,1.4) .. (10.93,3.29)   ;
%Curve Lines [id:da6088371312282304] 
\draw    (149.89,160) .. controls (125.04,172.61) and (114.83,165.74) .. (101.57,160.63) ;
\draw [shift={(99.89,160)}, rotate = 19.99] [color={rgb, 255:red, 0; green, 0; blue, 0 }  ][line width=0.75]    (10.93,-3.29) .. controls (6.95,-1.4) and (3.31,-0.3) .. (0,0) .. controls (3.31,0.3) and (6.95,1.4) .. (10.93,3.29)   ;
%Curve Lines [id:da6877559223914539] 
\draw    (159.78,150.11) .. controls (169.69,172.62) and (179.4,169.53) .. (192.75,156.55) ;
\draw [shift={(194,155.32)}, rotate = 135] [color={rgb, 255:red, 0; green, 0; blue, 0 }  ][line width=0.75]    (10.93,-3.29) .. controls (6.95,-1.4) and (3.31,-0.3) .. (0,0) .. controls (3.31,0.3) and (6.95,1.4) .. (10.93,3.29)   ;
%Curve Lines [id:da556224575924184] 
\draw    (241,150.11) .. controls (206.52,219.31) and (85.7,185.97) .. (51.4,161.13) ;
\draw [shift={(49.89,160)}, rotate = 37.83] [color={rgb, 255:red, 0; green, 0; blue, 0 }  ][line width=0.75]    (10.93,-3.29) .. controls (6.95,-1.4) and (3.31,-0.3) .. (0,0) .. controls (3.31,0.3) and (6.95,1.4) .. (10.93,3.29)   ;
%Curve Lines [id:da10667996859328843] 
\draw    (250.89,140.22) .. controls (233.27,129.84) and (230.09,118.71) .. (205.15,140.35) ;
\draw [shift={(204,141.37)}, rotate = 318.5] [color={rgb, 255:red, 0; green, 0; blue, 0 }  ][line width=0.75]    (10.93,-3.29) .. controls (6.95,-1.4) and (3.31,-0.3) .. (0,0) .. controls (3.31,0.3) and (6.95,1.4) .. (10.93,3.29)   ;
%Curve Lines [id:da8119294594365195] 
\draw    (299.89,140.22) .. controls (251.73,96.98) and (178.14,116.6) .. (151.09,139.18) ;
\draw [shift={(149.89,140.22)}, rotate = 318.5] [color={rgb, 255:red, 0; green, 0; blue, 0 }  ][line width=0.75]    (10.93,-3.29) .. controls (6.95,-1.4) and (3.31,-0.3) .. (0,0) .. controls (3.31,0.3) and (6.95,1.4) .. (10.93,3.29)   ;
%Curve Lines [id:da3772545076529191] 
\draw    (421,150.11) .. controls (350.08,113.87) and (334.46,117.21) .. (301.41,139.2) ;
\draw [shift={(299.89,140.22)}, rotate = 326.12] [color={rgb, 255:red, 0; green, 0; blue, 0 }  ][line width=0.75]    (10.93,-3.29) .. controls (6.95,-1.4) and (3.31,-0.3) .. (0,0) .. controls (3.31,0.3) and (6.95,1.4) .. (10.93,3.29)   ;
%Straight Lines [id:da9970085909591917] 
\draw    (440.78,150.11) -- (468,150.11) ;
\draw [shift={(470,150.11)}, rotate = 180] [color={rgb, 255:red, 0; green, 0; blue, 0 }  ][line width=0.75]    (10.93,-3.29) .. controls (6.95,-1.4) and (3.31,-0.3) .. (0,0) .. controls (3.31,0.3) and (6.95,1.4) .. (10.93,3.29)   ;
%Straight Lines [id:da6711652969360604] 
\draw    (489.78,150.11) -- (517,151.23) ;
\draw [shift={(519,151.32)}, rotate = 182.37] [color={rgb, 255:red, 0; green, 0; blue, 0 }  ][line width=0.75]    (10.93,-3.29) .. controls (6.95,-1.4) and (3.31,-0.3) .. (0,0) .. controls (3.31,0.3) and (6.95,1.4) .. (10.93,3.29)   ;
%Curve Lines [id:da38357717938949776] 
\draw    (479.89,140.22) .. controls (477.04,124.56) and (465.36,117.53) .. (432.41,139.21) ;
\draw [shift={(430.89,140.22)}, rotate = 326.12] [color={rgb, 255:red, 0; green, 0; blue, 0 }  ][line width=0.75]    (10.93,-3.29) .. controls (6.95,-1.4) and (3.31,-0.3) .. (0,0) .. controls (3.31,0.3) and (6.95,1.4) .. (10.93,3.29)   ;
%Curve Lines [id:da3652352768343291] 
\draw [line width=2.25]    (149.89,110.4) .. controls (282.33,73.5) and (326.06,82.13) .. (460.96,118.94) ;
\draw [shift={(463,119.5)}, rotate = 195.28] [color={rgb, 255:red, 0; green, 0; blue, 0 }  ][line width=2.25]    (17.49,-5.26) .. controls (11.12,-2.23) and (5.29,-0.48) .. (0,0) .. controls (5.29,0.48) and (11.12,2.23) .. (17.49,5.26)   ;
%Straight Lines [id:da48260581731869634] 
\draw    (309.78,150.11) -- (350,150.35) ;
\draw [shift={(352,150.37)}, rotate = 180.35] [color={rgb, 255:red, 0; green, 0; blue, 0 }  ][line width=0.75]    (10.93,-3.29) .. controls (6.95,-1.4) and (3.31,-0.3) .. (0,0) .. controls (3.31,0.3) and (6.95,1.4) .. (10.93,3.29)   ;

% Text Node
\draw (42,165) node [anchor=north west][inner sep=0.75pt]   [align=left] {$\displaystyle \sigma _{0}$};
% Text Node
\draw (91,164) node [anchor=north west][inner sep=0.75pt]   [align=left] {$\displaystyle \sigma _{1}$};
% Text Node
\draw (244,163) node [anchor=north west][inner sep=0.75pt]   [align=left] {$\displaystyle \sigma _{i}$};
% Text Node
\draw (421,166) node [anchor=north west][inner sep=0.75pt]   [align=left] {$\displaystyle \sigma _{n_{i}}$};
% Text Node
\draw (463,167) node [anchor=north west][inner sep=0.75pt]   [align=left] {$\displaystyle \sigma _{n_{i} +1}$};
% Text Node
\draw (360,150) node [anchor=north west][inner sep=0.75pt]   [align=left] {$\displaystyle \cdots $};
% Text Node
\draw (186,149) node [anchor=north west][inner sep=0.75pt]   [align=left] {$\displaystyle \cdots $};
% Text Node
\draw (140,165) node [anchor=north west][inner sep=0.75pt]   [align=left] {$\displaystyle \sigma _{2}$};
% Text Node
\draw (292,163) node [anchor=north west][inner sep=0.75pt]   [align=left] {$\displaystyle \sigma _{i+1}$};
% Text Node
\draw (521,149) node [anchor=north west][inner sep=0.75pt]   [align=left] {$\displaystyle \cdots $};
% Text Node
\draw (281,59) node [anchor=north west][inner sep=0.75pt]   [align=left] {$\displaystyle \leq \mathbf{1/2^{\mathnormal{i}}}$};

\end{tikzpicture}

\caption{\emph{An element $n_i$ in the diagonal sequence.} The states are ordered from left to right according to the enumeration $\sfEnum$; accordingly, $\sigma_0$, highlighted in green, is the terminal state. The arrows indicate probabilistic/nondeterministic/assignment transitions between the states. The yellow state space contains states indexed $\leq i$, and the red space contains states indexed $\geq n_i$. The probability of a run beginning from inside the yellow state space reaching the red space is $\leq 1/2^i$. Importantly, $n_i$ is the smallest such state index; meaning that if the red region included $\sigma_{n_i - 1}$, the aforementioned probability inequality will not hold.}
\label{fig:diagonal-sequence}
\VRS{Is this figure good? Another idea is to illustrate the diagonal sequence using a 1-D random walk.}
\end{figure}

To construct an unbounded supermartingale, one could consider the sum $\sum V_n$ varied across all $n \in \setOfNaturals$.
However, this sum could be $+\infty$ for certain states.
To combat this, we carefully choose an infinite subset of $\setOfNaturals$ to form the domain for $\sum V_n$.
Consider the sequence $(n_j)_{j \in \setOfNaturals}$ such that $n_j$ is the smallest number so that $R(i, n_j) \leq 2^{-j}$ for all $i \leq j$.
Each element in this sequence is certain to exist due to the monotonically non-increasing nature of $R(i, n)$ for fixed $i$ and the limit result of \cref{lem:scope-usefulness}.
Let's call this sequence $(n_j)_{j \in \setOfNaturals}$ the \emph{diagonal sequence}.
The operation of this diagonal sequence is illustrated in \cref{fig:diagonal-sequence}.
Restricting the domain of $\sum V_n$ to elements in this sequence will mean that no state is assigned $+\infty$ by the sum.
This is because for each $\sigma$, the values of $V_{n_j}(\sigma) = R(\sfEnum(\sigma), n_j)$ will certainly repeatedly halve after $j \geq \sfEnum(\sigma)$.
Further note that the supermartingale nature of the $V_n$ implies that this sum is also a supermartingale.
We thus have our required supermartingale
\begin{equation}
    \label{eq:construction-v}
    V(\sigma) = \sum_{j \in \setOfNaturals} V_{n_j}(\sigma) = \sum_{j \in \setOfNaturals} R(\sfEnum(\sigma), n_j)
\end{equation}
% Note that $V(\sigma) \leq \sfEnum(\sigma) + 1$

\begin{lemma}[Completeness]
    \label{lem:completeness-ast-martingale}
    \cref{proofrule:ast-martingale} is relatively complete.
\end{lemma}
\begin{proof}
    Take an $\AST$ $\CFG$ $\cG$, and set the inductive invariant $\Inv$
    %and $A$
    to $\reach(\cG)$. %and the singleton set containing the terminal state respectively.
    We first describe our choice for the variant function $U$.
    Since $\cG$ is $\AST$, for every $\sigma \in \reach(\cG)$, every scheduler must induce a finite path to a terminal state.
    \cref{cor:bound-on-shortest-term-runs} implies an upper bound on the length of the shortest terminal run from every $\sigma \in \Inv$.
    Set $U$ to map each $\sigma \in \Inv$ to this upper bound.

    If $U$ is bounded, setting $V(\sigma) = 1$ for all $\sigma \in \Inv$ suffices.
    Otherwise, set $V$ to the supermartingale function defined in \cref{eq:construction-v}.
    It is easy to observe that for every $r$, the sublevel set $V_{\leq r} = \{ \sigma \mid V(\sigma) < r\}$ is finite.
    This implies that $U$ is bounded within every sublevel set, and is hence compatible with this $V$.
    This completes the construction of the certificates required by the proof rule.

    We now argue that the invariant $\Inv$, the supermartingale $V$ and variant $U$ can each be represented in our assertion language of arithmetic interpreted over the rationals.
    We do this by encoding them first encoding them in the theory of natural numbers, and then using the relation $\Nat$ from \cref{theorem:robinson} to insert them into our assertion language.
    Recall that all computable relations can be encoded in $\ThQ$.
    We present techniques with which one can augment computable relations with first-order quantifiers to represent these entities.
    By doing so, we demonstrate that these sets are \emph{arithmetical}; see the works of \citet{Kozen06} and \citet{Rogers} for detailed accounts on arithmetical sets.
    
    % $A$ can trivially be represented in $\ThQ$.
    For representing $\Inv$ in $\ThQ$, consider the relation $I$ that contains tuples of the form $(k, \sigma_1, \sigma_2)$ where $k \in \setOfNaturals$ and $\sigma_1$ and $\sigma_2$ are states of $\cG$.
    Require $(k, \sigma_1, \sigma_2) \in I$ \emph{iff} there is a finite path of length $\leq k$ from $\sigma_1$ to $\sigma_2$.
    Clearly, $I$ is a computable relation and is thus representable in $\ThQ$.
    $\Inv(\sigma)$ can be represented from $I$ as $\exists k \cdot I(k, \sigma_0, \sigma)$ where $\sigma_0$ is the initial state of $\cG$.
    Similarly, the output of $U(\sigma)$ at every $\sigma$ can be represented using $I$ as $U(\sigma) = k \allowbreak \Longleftrightarrow \allowbreak I(\sigma, \sigma_\bot, k) \allowbreak \land \left(\forall n < k \cdot \lnot I\left(\sigma, \sigma_\bot, n\right)\right)$, where $\sigma_\bot$ is the terminal state.
    If $U$ were bounded, representing $V$ is trivial; we focus our attention on representing $V$ when $U$ isn't bounded.

    Representations of $R$ (defined in \cref{eq:scope-defn} and used to derive $V$) and $V$ are complicated slightly because they can output real numbers.
    Instead of capturing the precise values of these functions, we capture their Dedekind cuts instead.
    In other words, we show that the collections of rational numbers $\leq V(\sigma)$, $\geq V(\sigma)$, $\leq \varphi_i(n)$ and $\geq \varphi_i(n)$ are each representable for each $\sigma$, $i$, and $n$.

    The unrolling lemma implies that if the probability of termination is $p$, then for all $n \in \setOfNaturals$, assimilating a termination probability mass of at least $p - 1/n$ requires finitely many steps.
    It is simple to generalize this to observe that assimilating a probability mass of at least $R(i, n) - 1/n$ for the event $\Diamond\left( \left\{\sigma_m \in \Inv \mid m \geq n\right\}\right)$ when $\sigma_i$ is the initial state also requires finitely many steps.
    Furthermore, computing the probability of the occurrence of this event within $k$ steps is computable for every natural number $k$.
    These two facts indicate that lower bounds on $R(i, n)$ can be represented in $\ThQ$.
    Upper bounds on $R(i, n)$ can be represented by simply negating this lower bound representation.
    
    Using these, enable the representation of each member of the sequence $(n_j)_{j \in \setOfNaturals}$ that forms the domain of the sum that defines $V$.
    This enables representations of lower bounds on $V(\sigma)$, which in turn enables representations of upper bounds on $V(\sigma)$.
    Thus, the Dedekind cut of $V(\sigma)$ is representable in $\ThQ$.
    This completes the proof.
\end{proof}

\begin{example}[Random Walks]
    \label{ex:2drw-supermartingale}
For the 1DRW example from \cref{ex:srw}, we take $\Inv$ to be all program states% , $A$ to be the set containing the single terminal state $\{x_1 \coloneqq 0\}$,
and we set $V(x) = U(x) \coloneqq |x|$.
It's trivial to observe that all conditions required in \cref{proofrule:ast-martingale} are met, and therefore, the 1DRW is $\AST$.

Let us now consider the 2-D Random Walker (2DRW).
The $\AST$ nature of this program has been notoriously hard to prove using prior rules.
% The principal enabler of our rule on the 2DRW is its set $A$ that forms a circle around the origin.
We once again set $\Inv$ to the set of all states.
Our choice of distance variant $U$ assigns to each state $(x, y)$ in the plane to its Manhattan distance $|x| + |y|$ from the origin and terminal state $(0, 0)$.
Clearly, with probability at least $1/4$, $U$ reduces in each step across all states.
Now, define $V$ as
$$
V(x, y) \triangleq \sqrt{\ln \left(1 + ||x,y|| \right)}
$$
Where $||x,y||$ stands for the Euclidean distance $\sqrt{x^2 + y^2}$ of the point $(x, y)$ from the origin $(0, 0)$.
% \VRS{Should we hide the ugly differential? It's structure isn't an important part of the argument.}
\citet{MPWBook}  (and \citet[Section 2.3]{popov2DSRWBook}) proved that this function $V$ is a supermartingale for ``sufficiently large'' values of the norm $||x, y||$.
We now show that $V$ is a supermartingale at all states of the 2DRW.
Without any loss of generality, fix one of the program variables $y$.
Notice that $V$ is entirely a differentiable function.
% Take the partial differential of $V$ with respect to $x$:
% $$
% \frac{\delta V}{\delta x} = \frac{x}{2(1+ ||x, y||)\sqrt{||x,y||\ln(1+||x,y||)}}
% $$
It is also possible to show that the partial differential $\frac{\delta V}{\delta x}$ of $V$ with respect to $x$ is $< 1/2$ for all values of $|x| > 1$ at any fixed value of $y$.
To show this, first show that the partial double differential $\frac{\delta^2 V}{\delta x^2}$ is negative when $x > 0$ and positive otherwise. %TODO: What is this thing called? Google it once I get internet.
This means that $\delta V / \delta x$ always decreases as $|x|$ increases.
Furthermore, at $|x| = 1$, this differential is already $< 1/2$.

These two facts lets us conclude that $V(x, y) > \frac{1}{2} (V(x - 1, y) + V(x + 1, y))$.
Since $y$ was arbitrarily chosen to be fixed, we can swap $x$ and $y$ in the earlier paragraph to yield $V(x, y) > \frac{1}{2} (V(x, y-1) + V(x,y+1))$.
Putting these two inequalities together proves the supermartingale nature of $V$.
Finally, since the sublevel set $V_{\leq r}$ is finite, $V$ is compatible with $U$.
Therefore, the 2DRW is $\AST$.
\qed
\end{example}

\subsection{The Supermartingale-Variant rule}
\label{subsec:mciver-katoen-rule}
We now discuss a proof rule of \citet{McIverMKK18} that extends McIver and Morgan's variant \cref{proofrule:ast-variant}.
As the name suggests, the rule asks for a distance variant that is also a supermartingale.
% As noted earlier, this rule is quite similar to \cref{proofrule:ast-martingale}.

\begin{proofrule}[$\AST$ rule of \citet{McIverMKK18}]
    \label{proofrule:ast-mciver-katoen}
    To prove that $\cG$ is $\AST$, find
    \begin{enumerate}
    \item an inductive invariant $\Inv$ containing the initial state,
    \item two positive functions $p : \setOfPositiveReals \to (0, 1]$ and $d : \setOfPositiveReals \to \setOfPositiveReals$ such that $p$ and $d$ are \emph{antitone}, i.e., for all positive reals $r_1$ and $r_2$ with $r_1 \leq r_2$, $p(r_1) \geq p(r_2)$ and $d(r_1) \geq d(r_2)$, and
    \item a supermartingale variant function $V : \Inv \to \setOfReals$ that assigns $0$ to the terminal state and at all other states $(l, \bfx) \in \Inv$,
        \begin{enumerate}
            \item $V(l, \bfx) > 0$,
            \item $V(l, \bfx) > V(l_A, \bfx_A)$ for each $(l_A, \bfx_A) \in A$,
            \item if $(l, \bfx)$ is an assignment or nondeterministic state, then $V(l, \bfx) - V(l', \bfx') \geq d(V(l, \bfx))$ for all possible successor states $(l', \bfx')$,
            \item if $(l, \bfx)$ is a probabilistic state, then, over all successor states $(l', \bfx')$, $V(\sigma) \geq \sum \Pr(l, l') \allowbreak V(l', \bfx')$, and
            \item for all probabilistic states $(l, \bfx) \in \Inv$, the sum $\sum \sfPr(l, l')[\bfx] \geq p(V(l, \bfx))$ over all successor states $(l', \bfx')$ with $V(l, \bfx) - V(l', \bfx') \geq d(V(l, \bfx))$.
        \end{enumerate}
    \end{enumerate}
\end{proofrule}
This rule asks for a supermartingale function $V$ that is constrained by two antitone functions $p$ and $d$.
\citet{McIverMKK18} call $p$ and $d$ \emph{progress functions}, and the constraints on $V$ imposed by them \emph{progress conditions}.
In one step from each state $\sigma$ with $V(\sigma) = v$, the execution must reach a state $\sigma'$ with $V(\sigma') \leq v - d(v)$ with probability at least $p(v)$.
Furthermore, as $v$ reduces, $p$ and $d$ do not increase.
This means that this property of probably reducing $V$ in a single step quantitatively persists at $\sigma'$: in one step from $\sigma'$, the execution must reach a $V$ value of at most $V(\sigma') - d(v)$ with probability at least $p(v)$.
This is similar to the probable decrease requirement in the Variant \cref{proofrule:ast-variant}, with the supermartingale $V$ being the variant.
We thus call $V$ a \emph{supermartingale variant}.

We note that in their presentation, \citet{McIverMKK18} do not explicitly demand a supermartingale. 
They instead require that, for all real numbers $H$, the function $H \ominus V$ that maps states $\sigma \mapsto \max(H - V(\sigma), 0)$ is a \emph{submartingale}.
It is easy to see that this is equivalent to the supermartingale condition on $V$ in our presentation.

This rule is quite similar to our \cref{proofrule:ast-martingale}.
By combining the supermartingale and variant into one function, this rule neatly does away with the compatibility criteria required by our rule.
However, this means that their (supermartingale-)variants no longer map states to naturals.
This is why they require the antitone distance function $d$.
We show later that this $d$ may need to arbitrarily small for some states of certain programs, leading to more cryptic certificates.

Soundness of this rule follows from two facts.
First, if the execution were to be restricted to the sublevel sets $V^{\leq r} \triangleq \{ \sigma \mid V(\sigma) \leq r \}$, the variant nature of $V$ induced by the non-zero valuations of $p(r)$ and $d(r)$ imply almost-sure termination.
Second, because $V$ is a supermartingale, the probability with which an execution increases the value of $V$ to some high value $H$ is inversely proportional to $H$.
In fact, this probability goes to zero as $H \to \infty$.
Hence, the probability of increasing $V$ to infinity is zero, forcing executions to almost-surely remain within some sublevel set, within which the execution almost-surely terminates.

A subtle point in this rule is that the progress functions $p$ and $d$ must remain antitone and positive at reals outside the range of $V$.
\citet{McIverMKK18} require this to maintain soundness when $V$ is bounded, as in this case soundness arises from the variant nature of $V$, which demands non-zero valuations of $p$ and $d$ at the supremum of $V$.
Thus, replacing this with a requirement of a non-zero infimum on $p$ and $d$ across the range of $V$ also maintains soundness.

The completeness of this rule was open for quite some time.
In fact, it was conjectured that this rule was incomplete \cite{katoenETAPSTalk}.
However, we show that our techniques can be extended to prove the completeness of their rule as well.
\begin{theorem}
    \cref{proofrule:ast-mciver-katoen} is sound and relatively complete.
\end{theorem}
Soundness follows from the intuitions provided earlier.
Completeness results from the following argument.
Take the supermartingale $V$ constructed in \cref{eq:construction-v} in the completeness argument of \cref{proofrule:ast-martingale}.
If this $V$ were one-one, it would immediately satisfy the conditions required by \cref{proofrule:ast-mciver-katoen}.
We thus detail a procedure to transform $V$ into a one-one supermartingale.
We start by injecting some ``gaps'' into $V$, by making it strictly decrease in expectation at certain states.
We do so by taking the natural logarithm of the output of $V$; the weighted AM-GM inequality gives us the gaps.
We then mildly perturb the values $V$ now assigns to certain states in a way that doesn't affect its supermartingale nature, but does make $V$ one-one.
This is generally not possible; the gaps are essential for this perturbation to work.
We present the full argument for completeness in \cref{sec:completeness-mciver-katoen-rule}.
% We present the full argument for completeness in our full version.

\citet{McIverMKK18} demonstrated the applicability of \cref{proofrule:ast-mciver-katoen} for a variety of programs.
% They further noted the theoretical applicability of their rule for the 2D random walker.
For example, its easy to see that the supermartingale discussed in \cref{ex:2drw-supermartingale} acts sufficiently as a supermartingale variant.
% One can show that the progress functions defined below match the requirements of the rule.
% We show, in \cref{sec:mciver-katoen-2drw}, that the progress functions defined below match the requirements of the rule.
It is possible to show that the progress functions defined below match the requirements of the rule.
$$
p(v) = \frac{1}{4} \quad\quad\quad\quad d(v) = \frac{1}{2v} \left( v^2 - \ln\left(e^{v^2} - 1\right) \right)
$$
For certain programs however, the certificates that match this rule are necessarily contrived.
% This is exemplified by the long tail end.

\begin{figure}[t]
    %\begin{minipage}[c]{0.5\linewidth}
    \small
    \begin{subfigure}{0.59\textwidth}
    \begin{lstlisting}[language=python, mathescape=true, escapechar=|, xleftmargin=15pt]
x, y $\coloneqq$ 2, 1
while (x > 0):
  if (y = 1): y $\coloneqq$ 0 |$\pChoice{1/2}$| x $\coloneqq$ 2^x |\label{line:tel-loop-1}|
  else: x $\coloneqq$ x - 1 |\label{line:tel-loop-2}|
    \end{lstlisting}
    \caption{Code view}
    \end{subfigure}
    %\end{minipage}
    %\quad\quad\quad\quad
    %\begin{minipage}[c]{0.3\linewidth}
    \begin{subfigure}{0.4\textwidth}
    \input{tail-end-loop.tikz}
    \caption{Transition system view}
    \end{subfigure}
    %\end{minipage}
    \caption{The long tail end, an $\AST$ program for which synthesizing a supermartingale variant is difficult. We refer to the second loop at \listingLineRef{line:tel-loop-2}, which corresponds to the right-half of the transition system, as the tail end loop of the program. The final/terminal state is shown in \textcolor{red}{red}.}
    \label{code:mciver-morgan-ctrex}
    % \finishcodefigure
\end{figure}
\begin{example}[The long tail end]
    \label{ex:tail-end-loop}
    This program, described in \cref{code:mciver-morgan-ctrex}, initializes its variable $\ttx$ to $2$ and has two phases.
    In the first phase, the program either updates the variable $\ttx$ to $2^\ttx$ or moves onto the second phase with equal probability.
    Note that, while we do not explicitly include exponentiation in our program model, it is simple enough to consider this a part of the syntax available to the programmer to build terms.
    In the second phase, which we call the tail loop, the program deterministically decrements $\ttx$ until it hits $0$, at which point it terminates.
    It's quite easy to see that this program is $\AST$.
    
    It isn't difficult to come up with separate distance variants $U$ and supermartingales $V$ to apply our \cref{proofrule:ast-martingale} on this program.
    The distance variant simply maps each program state at line \listingLineRef{line:tel-loop-1} to $\ttx + 1$ and states at \listingLineRef{line:tel-loop-2} to $\ttx$.
    Since the distance variant is unbounded, the supermartingale $V$ must be unbounded too.
    Start by setting $V$ to $2$ at the beginning of the execution.
    Each time the probabilistic operator at \listingLineRef{line:tel-loop-1} resolves to increasing $\ttx$, add $1$ to $V$.
    Subtract $1$ from $V$ otherwise.
    Make no change to $V$ during the execution of \listingLineRef{line:tel-loop-2}.
    As $V$ increases with $\ttx$, the variant $U$ is bounded when restricted to the sublevel sets where $V \leq r$.
    
    However, producing a supermartingale-variant $V$ in accordance to \cref{proofrule:ast-mciver-katoen} for this program is more difficult.
    This is because the probability distribution at \listingLineRef{line:tel-loop-1} means that $V$ can at-most double in each round
    However, the distance to termination, measured by the variable $\ttx$, grows exponentially.
    Thus, a very slow decrease of $V$ must be instituted during the execution of the tail loop at \listingLineRef{line:tel-loop-2}.
\end{example}

%% file: figures/tail-end-loop.tikz
\begin{tikzpicture}
	\begin{pgfonlayer}{nodelayer}
		\node [style=plain circle] (0) at (0, 0) {};
		\node [style=red circle] (3) at (7, 0) {};
		\node [style=plain circle] (4) at (4, 0) {};
		\node [style=none] (6) at (4, 1.25) {\scriptsize{$x := x - 1$}};
		\node [style=none] (7) at (2, -0.75) {\scriptsize{$y \neq 1?$}};
		\node [style=none] (8) at (5.5, -0.75) {\scriptsize{$x \leq 0?$}};
		\node [style=plain circle] (9) at (0, 2) {};
		\node [style=none] (10) at (-1.5, 1) {\scriptsize{$x := 2^x$}};
		\node [style=none] (11) at (1.5, 1) {\scriptsize{$y = 1?$}};
		\node [style=none] (12) at (-2, 0) {};
		\node [style=none] (13) at (-1.5, -0.5) {\scriptsize{$x := 2$}};
		\node [style=none] (14) at (-1.5, -1) {\scriptsize{$y := 1$}};
	\end{pgfonlayer}
	\begin{pgfonlayer}{edgelayer}
		\draw [style=direction edge] (0) to (4);
		\draw [style=direction edge] (4) to (3);
		\draw [style=direction edge, in=135, out=45, loop] (4) to ();
		\draw [style=direction edge, bend right=45, looseness=1.25] (0) to (9);
		\draw [style=direction edge, bend right=45, looseness=1.25] (9) to (0);
		\draw [style=direction edge] (12.center) to (0);
	\end{pgfonlayer}
\end{tikzpicture}

%% file: quant.tex
\section{Quantitative Termination}
\label{sec:quant}

We now present proof rules to reason about \emph{quantitative} notions of termination.
These rules are meant to prove lower and upper bounds on the probability of termination.
% These are \emph{quantitative} notions of termination.
% This is in contrast to \emph{qualitative termination}, the nomenclature employed for $\AST$.
As mentioned earlier, the proof rules we specify in this section can be used to show irrational bounds on the termination probability;
this is a simple consequence of the fact that all possible termination probabilities can be expressed in our assertion language.

Our rules make essential use of the stochastic invariants of \citet{ChatterjeeNZ17}.
\begin{definition}[Stochastic Invariants \cite{ChatterjeeNZ17}]
    \label{def:stochastic-invariants}
    Let $\cG=(L, V, \allowbreak l_{init}, \bfx_{init}, \ordMapsTo, G, \sfPr, \sfUp)$ be a $\iCFG$.
    Suppose $\SIpred$ is a subset of states and let $p$ be a probability value. 
    The tuple $(\SIpred, p)$ is a \emph{stochastic invariant} (SI) if, under any scheduler $\sched$, the probability mass of the collection of runs beginning from $(l_{\init}, \bfx_{\init})$ leaving $\SIpred$ is bounded above by $p$, i.e., 
    $$
    \sup{}_\sched \Prob{}_\sched\left[ \rho \in \runs_\cG \mid \exists n \in \setOfNaturals \cdot \rho[n] \not\in \SIpred  \right] \leq p 
    $$
\end{definition}

Intuitively, stochastic invariants generalize the standard notion of invariants to the probabilistic setting.
Given a stochastic invariant $(\SIpred, p)$, the program execution is expected to hold $\SIpred$ (i.e., remain inside $\SIpred$) with probability $\geq 1-p$.
As with invariants, the collection of states in stochastic invariants is typically captured by a predicate written in the assertion language of the program logic.
In this work however, we do not characterize stochastic invariants directly. We instead use stochastic invariant indicators.
\begin{definition}[Stochastic Invariant Indicator \cite{ChatterjeeGMZ22}]
    Let $\cG$ be the $\CFG$ $(L, V, l_{init}, \allowbreak \bfx_{init}, \ordMapsTo, G, \sfPr, \sfUp)$. 
    A tuple $(\SI, p)$ is a \emph{stochastic invariant indicator} (SI-indicator) if $p$ is a probability value and 
$\SI : L \times \setOfIntegers^{V} \to \setOfReals$ is a partial function such that $\SI(l_{init}, \bfx_{init}) \leq p$, and 
for all states $(l, \bfx)$ reachable from $(l_{\init}, \bfx_{\init})$, 
    \begin{enumerate}
        \item $\SI(l, \bfx) \geq 0$.
        \item if $(l, \bfx)$ is an assignment or nondeterministic state, then $\SI(l, \bfx) \geq \SI(l', \bfx')$ for every successor $(l', \bfx')$.
        \item if $(l, \bfx)$ is a probabilistic state, then $\SI(l, \bfx) \geq \sum \sfPr((l, l'))[\bfx] \times \SI(l', \bfx')$ over all possible successor states $(l', \bfx')$.
    \end{enumerate}
\end{definition}
Observe that functions $\SI$ in the SI-indicators are supermartingale functions.
These $\SI$ are typically most interesting at states $\sigma$ where $\SI(\sigma) < 1$; in fact, the collection of states $\sigma$ with this property corresponds to an underlying stochastic invariant with the same probability value as the SI-indicator.
This was formally proven by \citet{ChatterjeeGMZ22}.
\begin{lemma}[\cite{ChatterjeeGMZ22}]
    \label{lem:si-sii-equiv}
    Let $\cG$ be an $\iCFG$.
    For each stochastic invariant $(\SIpred, p)$ of $\cG$, there exists a stochastic invariant indicator $(\SI, p)$ of $\cG$ such that $\SIpred \supseteq \{ \gamma \in \Sigma_\cG \mid \SI(\gamma) < 1 \}$.
    Furthermore, for each stochastic invariant indicator $(\SI, p)$ of $\cG$, there is a stochastic invariant $(\SIpred, p)$ such that $\SIpred = \{ \gamma \in \Sigma_\cG \mid \SI(\gamma) < 1 \}$.
\end{lemma}
The SI-indicator $\SI$ corresponding to the stochastic invariant $\SIpred$ maps each state $\sigma$ the probability with which runs beginning from $\sigma$ exit $\SIpred$.
Thus, the SI-indicator tracks the probability of violating the stochastic invariant.
Observe that $\SI(\sigma) \geq 1$ for all states $\sigma \not\in \SIpred$.

We will use SI-indicators in our proof rules.
% Note that we cannot use a single stochastic invariant: this is too weak to ensure soundness.
% Instead, we use a family of stochastic invariants, one for each reachable state.
Unlike stochastic invariants, representing SI-indicators as assertions is more complicated.
Because SI-indicators map states to reals, we cannot always write them directly in our language of assertions.
When our proof rules apply, however, we will demonstrate how one can find expressions in our assertion language to represent the SI-indicators our proof rule needs.

Our upper bound rule is immediate from the nature of stochastic invariants.
\begin{proofrule}[Upper Bounds Rule]
    \label{proofrule:SI-upper-bounds}
    To show that $\termProb(\cG) \leq p$ for $0 < p < 1$, find
    \begin{enumerate}
        \item an inductive invariant $\Inv$ containing $(l_{\init}, \bfx_{\init})$, 
        \item a function $\SI : \Inv \to \setOfReals$ such that $\SI(l_{init}, \bfx_{init}) \leq p$ and for all $(l, \bfx) \in \Inv$,
        \begin{enumerate}
            \item $\SI(l, \bfx) \geq 0$;
            \item if $(l, \bfx)$ is an assignment, or nondeterministic state, then $\SI(l, \bfx) \geq \SI(l', \bfx')$ for every successor $(l', \bfx')$;
            \item if $(l, \bfx)$ is a probabilistic state, then $\SI(l, \bfx) \geq \allowbreak \sum \sfPr((l, \allowbreak l'))[\bfx] \times \SI(l', \bfx')$ over all possible successor states $(l', \bfx')$.
            \item if $(l, \bfx)$ is a terminal state, then $\SI(l, \bfx) \geq 1$.
        \end{enumerate}
    \end{enumerate}
\end{proofrule}
This rule asks for an SI-indicator (and therefore a stochastic invariant) that excludes the terminal state.
Therefore, the probability of termination is the probability of escaping the SI-indicator, which is included in the property of the indicator.
Notice that, because the SI-indicator is a supermartingale function, the mere existence of a bounded supermartingale that assigns to the terminal state a value $\geq 1$ is sufficient to extract an upper bound.
  
\begin{lemma}
    \label{lem:SI-upper-bounds}
    \cref{proofrule:SI-upper-bounds} is sound and relatively complete. 
\end{lemma}
  
\begin{proof}
    \cref{lem:si-sii-equiv} shows that the pair $(\{\gamma \in \Inv \mid \SI(\gamma) < 1\}, p)$ is a stochastic invariant of $\cG$.
    Since $\SI(\gamma) \geq 1$ at the terminal state, this stochastic invariant doesn't contain the terminal state.
    Soundness of the rule trivially follows from the fact that, in order to terminate, a run must leave this invariant and this probability is bounded above by $p$.

    For completeness, set $\Inv = \reach(\cG)$ and let $\sigma_\bot \in \Inv$ be the terminal state.
    Set $\SIpred = \Inv \setminus \{ \sigma_bot \}$, and observe that the collection of runs leaving $\SIpred$ is identical to the collection of terminal runs.
    Therefore, $(\SIpred, p)$ must be a stochastic invariant.
    \cref{lem:si-sii-equiv} indicates the existence of the SI-Invariant $\SI$ from $\SIpred$.
    $\SI$ immediately satisfies the conditions of the rule.

    To represent $\SI$ in our assertion language, note again that $\SI(\gamma)$ is precisely the probability of termination from $\gamma$.
    This probability is lower bounded by the probability of termination within some $m$ steps from $\gamma$.
    The latter probability is computable, and is hence representable in $\ThQ$.
    Prepending an appropriate universal quantifier for $m$ allows us to form lower bounds for $\SI$ in $\ThQ$.
    Our upper bound rule is thus relatively complete.
\end{proof}

\subsection{Towards a Lower Bound}

For a lower bound on the probability of termination, we start with the following rule from \citet{ChatterjeeGMZ22}. 

\begin{proofrule}[SI-indicators for Lower Bounds]
    \label{proofrule:SI-simple-lower-bounds-rule}
    To prove that $\termProb(\cG) \geq 1 - p$ for $0 < p < 1$, find
    \begin{enumerate}
        \item an inductive invariant $\Inv$ containing $(l_{\init}, \bfx_{\init})$, 
        \item a stochastic invariant indicator $\SI : \Inv \to \setOfReals$ such that $\SI(l_{\init}, \bfx_{\init}) \geq p$, and for all $(l, \bfx) \in \Inv$,
        \begin{enumerate}
            \item $\SI(l, \bfx) \geq 0$.
            \item if $(l, \bfx)$ is an assignment, or nondeterministic state, then $\SI(l, \bfx) \geq \SI(l', \bfx')$ for every successor $(l', \bfx')$.
            \item if $(l, \bfx)$ is a probabilistic state, then $\SI(l, \bfx) \geq \sum \sfPr((l, \allowbreak l'))[\bfx] \allowbreak \times \SI(l', \bfx')$ over all possible successor states $(l', \bfx')$.
        \end{enumerate}
        \item an $\epsilon > 0$,
        \item a variant $U : \Inv \to \setOfNaturals$ that is bounded above by some $H$, maps all states $\{\gamma \in \Inv \mid \SI(\gamma) \geq 1 \lor \gamma \text{ is terminal} \}$ to $0$, and for other states $(l, \bfx) \in \Inv$,
        \begin{enumerate}
            \item if $(l, \bfx)$ is an assignment, or nondeterministic state, $U(l', \bfx') < U(l, \bfx)$ for every successor $(l', \bfx')$.
            \item if $(l, \bfx)$ is a probabilistic state, $\sum \sfPr(l, l')[\bfx] > \epsilon$ over all successor states $(l', \bfx')$ with $U(l', \bfx') < U(l, \bfx)$.
        \end{enumerate}
    \end{enumerate}
\end{proofrule}
This rule demands a SI-indicator $(\SI, p)$ and a bounded variant $U$ such that $\SI$ induces a stochastic invariant $(\SIpred, p)$ that contains the terminal state.
Intuitively, this rule works by splitting the state space into terminating and possibly non-terminating segments.
The invariant $\SIpred$ represents the terminating section of the state space.
The application of McIver \& Morgan's \cref{proofrule:ast-variant} with the variant $U$ allows us to deduce that, were the execution be restricted to $\SIpred$, the program almost-surely terminates.
Observe that the variant $U$ effectively considers all states outside $\SIpred$ to be terminal.
Therefore, $\cG$ almost-surely either escapes $\SIpred$ or terminates within $\SIpred$.
Non-termination is thus subsumed by the event of escaping the invariant, the probability of which is $p$.
Hence, the probability of termination is $\geq 1 - p$.

% TODO: I'm not sure what the point of this paragraph is, but I think there's a comparison to be made here.
Observe that, like our $\AST$ \cref{proofrule:ast-martingale}, this rule requires a supermartingale and a variant function that work in tandem.
However, unlike \cref{proofrule:ast-martingale}, the supermartingale and the variant are entirely bounded.

This soundness argument was formally shown by \citet{ChatterjeeGMZ22}.
In their original presentation, they do not specify an exact technique for determining the almost-sure property of either termination within or escape from the induced stochastic invariant $\SIpred$.
We will explain our choice of the bounded variant \cref{proofrule:ast-variant} of \citet{McIverMorganBook} in a moment.
They additionally claimed the completeness of this rule, assuming the usage of a complete rule for almost-sure termination.
If their completeness argument were true, it would indicate that all probabilistic programs induce state spaces that can neatly be partitioned into terminating and non-terminating sections.
In \cref{sec:ctrex}, we show that this isn't the case using a counterexample at which this split isn't possible.

Nevertheless, this rule is complete for finite-state programs.
This finite-state completeness pairs well with the finite-state completeness of the bounded variant \cref{proofrule:ast-variant} we use to certify the almost-sure property contained in the rule.

\begin{figure}%[t]{r}{0.3\linewidth}
    \input{loop-exploitation.tikz}
    \caption{If shortest runs of $\Sigma_{good}$ were too long. \textnormal{This is a representation of the collection of executions beginning at a good state $\sigma \in \Sigma_{good}$, according to the partitioning system suggested in the proof of \cref{lem:SI-lower-bounds-partial-completeness}. The black nodes are the terminal states; they all lead to the single terminal state. The \textcolor{blue}{blue} states are identical to each other; the same holds for the \textcolor{blue-green}{blue green} states. The pathological scheduler $\sched'$ always takes the \textcolor{red}{red} back edges, rendering no terminal runs from $\sigma$.}}
    \label{fig:loop-exploitation}
\end{figure}

\begin{lemma}
    \label{lem:SI-lower-bounds-partial-completeness}
    \cref{proofrule:SI-simple-lower-bounds-rule} is complete for finite state $\CFG$s. 
\end{lemma}
\begin{proof}
    In this proof, we argue about the stochastic invariants directly.
    The SI-indicator and variant functions can be derived from them using prior techniques \cite{ChatterjeeGMZ22, McIverMorganBook}.

    Let $\cG$ be a finite state $\iCFG$.
    Partition $\reach(\cG)$, the set of states reachable from the initial state of $\cG$, into 
    \begin{inlinelist}
        \item the singleton containing the terminal state $\Sigma_\bot$, 
        \item the \emph{bad} states $\Sigma_{bad}$ with the property that no finite paths ending on these states can be extended to a terminal run, 
        \item the \emph{good} states $\Sigma_{good}$ such that if a scheduler $f$ induces a finite path ending at a good state $\sigma$, $f$ induces at least one terminating run passing through $\sigma$, and 
        \item the remaining \emph{neutral} states $\Sigma_{neutral}$, with the property that each neutral state $\sigma$ is associated with a pathological scheduler $\sched_\sigma$ that induces runs that, if they pass through $\sigma$, do not terminate. 
    \end{inlinelist}
    Notice that, as long as $p < 1$ (the case where $p = 1$ is trivial, so we skip it), the initial state of $\cG$ is in $\Sigma_{good}$. 
    We show that $(\Sigma_{good} \cup \Sigma_\bot, p)$ is the required stochastic invariant. 
    
    % NOTE: This analysis isn't necessary; we can show that $\termProb(\cG(\sigma)) > 0$ for all good states $\sigma \in \Sigma_{good}$. With the corollary of the required simulation time, this result becomes obvious. However, I'm keeping it for now, because the pumping argument is kinda cute.
    We begin by showing that runs that remain inside $\Sigma_{good} \cup \Sigma_\bot$ almost-surely terminate. 
    Fix a state $\sigma \in \Sigma_{good}$. 
    Map to every scheduler $\sched$ that induces runs passing through $\sigma$ the shortest terminating consistent finite path $\pi_\sched$ beginning from $\sigma$. 
    Suppose, for some scheduler $\sched$, $|\pi_\sched| > |\reach(\cG)|$.
    Then, $\pi_\sched$ must visit some $\sigma' \in \reach(\cG)$ twice. 
    This indicates a loop from $\sigma' \to \sigma'$ in $\cG$, and there must thus exist a scheduler $\sched'$ that extends $\pi_\sched$ by repeating this loop infinitely often. 
    Since $\pi_\sched$ is the smallest terminating finite path beginning from $\sigma$, the same holds for all terminating finite paths consistent with $\sched$ beginning from $\sigma$. 
    Thus, all terminating runs consistent with $\sched$ that pass through $\sigma$ must contain a loop. 
    There must exist a scheduler $\sched'$ which exploits these loops and yields no terminating runs passing through $\sigma$. 
    \cref{fig:loop-exploitation} depicts the operation of $\sched'$.
    % We leave the details of the construction of $\sched'$ to the diligent reader. \RM{reviewers get pissed because they are not "diligent readers" and find it condescending ;-)} \VRS{DONE: Add figure here explaining the bad scheduler $\sched'$, and reword this}
    Observe that the existence of $\sched'$ contradicts $\sigma \in \Sigma_{good}$. 

    Therefore, from every state $\sigma \in \Sigma_{good}$, there is a terminating run of length $\leq |\reach(\cG)|$ no matter which scheduler is used. 
    Thus, the probability of leaving $\Sigma_{good}$ from $\sigma$ is bounded below by $q^{|\reach(\cG)|}$, where $q$ is the smallest transition probability of $\cG$
    (note that $q$ only exists because $\cG$ is finite state). 
    Further, $\Sigma_{good}$ only contains non-terminal states. 
    This enables the zero-one law of probabilistic processes \cite[Lemma 2.6.1]{McIverMorganBook}, allowing us to deduce the almost-certain escape from $\Sigma_{good}$ to $\Sigma_\bot \cup \Sigma_{bad} \cup \Sigma_{neutral}$. % I think this is simpler than attempting to re-prove the zero-one law here. 
    Hence, under all schedulers, the probability of either terminating inside $\Sigma_{good} \cup \Sigma_\bot$ or entering $\Sigma_{bad} \cup \Sigma_{neutral}$ is $1$. 

    % This paragraph describes why $\Sigma_{neutral}$ shouldn't be a part of the stochastic invariant. 
    We now show that $(\Sigma_{good} \cup \Sigma_{bot}, p)$ is a stochastic invariant. 
    It's easy to see that if a run ever enters $\Sigma_{bad}$, it never terminates. 
    We know that if an execution enters some $\sigma \in \Sigma_{neutral}$ under a pathological scheduler $\sched_\sigma$, it never terminates. 
    Let $\sigma_1$ and $\sigma_2$ be neutral states and $\sched_1$ and $\sched_2$ be their corresponding pathological schedulers. 
    % If, under $f_1$, a run passes through $\sigma_1$, it never terminates. 
    Notice that $\sched_1$ may induce terminating executions that pass through $\sigma_2$. 
    One can build a scheduler $\sched_3$ that mimics $\sched_1$ until the execution reaches $\sigma_2$, and once it does, mimics $\sched_2$. 
    Thus, $\sched_3$ would produce the pathological behaviour of both $\sched_1$ and $\sched_2$. 
    In this way, we compose the pathological behavior of all neutral states to produce a scheduler $\frakf$ that induces runs that, if they enter a neutral state, never terminate. 
    Under $\frakf$, leaving $\Sigma_{good} \cup \Sigma_\bot$ is equivalent to non-termination, and all terminating runs are made up of good states until their final states. 

    % A cool result: you can only move from good states to neutral states only on probabilistic transitions. 
    % Additionally, you can only move from neutral states to good states on nondeterministic transitions. 

    Take a scheduler $\sched$ that induces terminating runs that pass through neutral states. 
    Compose the scheduler $\sched'$ that mimics $\sched$ until the execution enters a neutral state and mimics $\frakf$ from then on. 
    Let $T_\sched$ and $T_{\sched'}$ be the collection of terminating runs consistent with $\sched$ and $\sched'$ respectively. 
    Each run in $T_{\sched'}$ is made up of good and/or terminal states, and is therefore consistent with $\sched$. 
    Hence, $T_\sched \supset T_{\sched'}$ and, because $\sched$ and $\sched'$ agree on $T_{\sched'}$, we have $\Prob_\sched(T_\sched) > \Prob_\sched(T_{\sched'})$, 
	where $\Prob_\sched$ is the probability measures in the semantics of $\cG$ induced by $\sched$. 
    
    Leaving $\Sigma_{good} \cup \Sigma_\bot$ is equivalent to entering $\Sigma_{neutral} \cup \Sigma_\bot$. 
    Observe that the probability of never leaving $\Sigma_{good} \cup \Sigma_\bot$ under $\sched$ is the same as the probability of never leaving $\Sigma_{good} \cup \Sigma_\bot$ under $\sched'$, as $\sched$ and $\sched'$ agree until then. 
    Furthermore, the probability measure of never leaving $\Sigma_{good} \cup \Sigma_\bot$ is just $\Pr_\sched(T_{\sched'}) = \Pr_{\sched'}(T_{\sched'})$. 
    Add $\termProb(\cG) \geq 1-p \implies p_{\sched'}(T_{\sched'}) \geq 1-p$, and we get that, under $\sched$, the probability of leaving $\Sigma_{good} \cup \Sigma_\bot$ is upper bounded by $p$. 
    Since this is true for any $\sched$, the lemma is proved. 
\end{proof}
We do not show the relative completeness of this rule for finite-state programs; nevertheless, this is easy to show using techniques discussed in prior rules.

\subsection{Our Rule}

We now show a sound and complete rule for lower bounds that fixes the prior \cref{proofrule:SI-simple-lower-bounds-rule}.
This rule is implicitly contained in the details of the erroneous proof of \cite{ChatterjeeGMZ22}.
% It is principally similar to \cref{proofrule:ast-SI}, in that it identifies finite sub-instances where prior rules can apply.
% It then combines the proofs of these sub-instances to deduce the desired lower bound.
Principally, this rule exploits the fact that establishing lower bounds strictly under the exact termination probability can be done using the prior \cref{proofrule:SI-simple-lower-bounds-rule}.
This is because the program is guaranteed to exceed these non-strict lower bounds before leaving a finite subset of its state space.
Recall the unrolling lemma from \cref{subsec:unrolling-lemma}.
The unrolling lemma implies that if a $\CFG$ $\cG$ terminates with probability $p$, then for every $\epsilon > 0$, there is a finite unrolling of the execution of $\cG$ which establishes a termination probability of $p - \epsilon$.
This finite unrolling is tantamount to a finite-state subprogram of $\cG$, over which the prior \cref{proofrule:SI-simple-lower-bounds-rule} operates.

\begin{proofrule}[Lower Bounds Rule]
    \label{proofrule:lower-bounds-ours}
    To show that $\termProb(\cG) \geq 1 - p$ for $0 < p < 1$, find, for all $n \in \setOfNaturals$, functions $SI_n$ and $U_n$ that enable the application of \cref{proofrule:SI-simple-lower-bounds-rule} to deduce $\termProb(\cG) \geq 1 - \left(p + \frac{1}{n}\right)$.
\end{proofrule}

Soundness of this rule follows trivially from the soundness of the prior \cref{proofrule:SI-simple-lower-bounds-rule}.
The completeness of this rule is derived from the unrolling lemma; to reach a termination probability of $p$, the program must be able to amass a termination probability of $p - \frac{1}{n}$ within a finite subspace.
\cref{lem:SI-lower-bounds-partial-completeness} shows that finiteness can always be captured by the prior \cref{proofrule:SI-simple-lower-bounds-rule}.

\begin{lemma}
    \label{lem:lower-bounds-completeness}
    \cref{proofrule:lower-bounds-ours} is relatively complete.
\end{lemma}
\begin{proof}
    Let $\cG$ be $\CFG$ such that $\termProb(\cG) \geq 1 - p$ for some $p > 0$. 
    Fix some $n \in \setOfNaturals$. 
    Let $k_n$ be the upper bound guaranteed by the unrolling lemma over the required simulation times across all schedulers to amass a termination probability of $1 - (p + 1/n)$. 
    Denote by $\Sigma_n$ the set of states $\sigma$ such that there is a finite path of length at most $k_n$ beginning at $(l_{init}, \bfx_{init})$ and ending at $\sigma$. 
    Clearly, $\Sigma_n$ must be finite and, for any scheduler $\sched$, the probability measure of the collection of terminating runs made up of states in $\Sigma_n$ consistent with $\sched$ must be $\geq 1 - (p + 1/n)$. 
    The finite-state completeness of \cref{proofrule:SI-simple-lower-bounds-rule} implies the existence of an SI-indicator $(\SI_n, p + 1/n)$ and a variant $U_n$ that certifies that the restricted program $\cG$ constrained to states within $\Sigma_n$ terminates with probability $\geq 1 - (p + 1/n)$.

    Let us now show how we can represent these SI-indicators and variants in our assertion language.
    As in the proof of the completeness of the martingale \cref{proofrule:ast-martingale}, consider the relation $I$ such that $(k, \sigma_1, \sigma_2) \in I$ \emph{iff} there is a finite path of length $\leq k$ from $\sigma_1$ to $\sigma_2$ in $\cG$.
    Clearly, $I$ is a computable relation and is thus representable in $\ThQ$.
    $\Inv$ and each $U_n$ can easily be represented in $\ThQ$ using this relation $I$ in exactly the same way as with \cref{proofrule:ast-martingale}.

    Let us now detail techniques to represent the SI-indicators $\SI_n$ in our assertion language.
    Denote by $(\SIpred_n, p+1/n)$ the stochastic invariant associated with $(\SI_n, p + 1/n)$ as required by \cref{lem:si-sii-equiv}.
    \citet{ChatterjeeGMZ22} proved that for every state $\gamma$, $\SI_n(\gamma)$ is precisely the probability with which executions beginning at $\gamma$ leave the set $\SIpred$.
    \citet{MS24arxiv} showed how lower bounds on this probability can be encoded in $\ThQ$.
    This encoding makes essential use of the unrolling lemma.
    Intuitively, the probability with which executions beginning at $\gamma$ escape $\SIpred_n$ within some $m$ steps always lower bounds $\SI(\gamma)$.
    The former probability is computable.
    The unrolling lemma ensures that in spite of non-determinism, augmenting $m$ with a universal quantifier is exactly $\SI(\gamma)$.
    Hence, this rule is relatively complete.
%
    % NOTE: This will be useful when/if we present guards as a convenience mechanism for the rule. 
    % We overapproximate $\Sigma_n$ using a predicate $\varphi_n$. 
    % Since $\Sigma_n$ is finite, it induces lower and upper bounds over the values taken by each variable $x \in V$. 
    % Let $a_x$ and $b_x$ be the lowest and highest values taken by the variable $x$ over all states in $\Sigma_n$. 
    % $\Sigma_n$ can also induce restrictions over the program locations visited. %TODO: I don't think this is necessary. I'll keep this here for now. 
    % Let $L'$ be the collection of program locations contained in the states visited by $\Sigma_n$. 
    % Construct $\varphi_n$ as 
    % $$
    % \varphi_n = (l \in L') \land \bigwedge_{x \in V} (a_x \leq x \leq b_x) 
    % $$
    % $\varphi_n$ is a guard predicate, as there are finitely many variables $V$ and finitely many locations $L'$. 
    % Let $\Sigma_{\varphi_n}$ be the collection of states that evaluate $\varphi_n$ to $\top$. 
    % Observe that, since the variables range over the integers, $\Sigma_{\varphi_n}$ is also finite. 
    % Furthermore, $\Sigma_{\varphi_n} \supseteq \Sigma_n$. 
%
    % Construct the $\iCFG$ $\cG_{\varphi_n}$ according to the procedure described in \cref{proofrule:katoen-guard-rulerule}. %TODO: Mark the transformation differently, somehow 
    % Let $T_{\varphi_n}$ be the collection of terminating runs of $\cG_{\varphi_n}$. 
    % Observe that $T_{\varphi_n}$ must contain all terminating runs in $\Sigma_n$. 
    % Therefore, for any scheduler $f$, the probability measure ascribed to the collection of consistent runs in $T_{\varphi_n}$ must be $\geq p - 1/n$. 
  \end{proof}

\subsection{Counterexample to Completeness for \cref{proofrule:SI-simple-lower-bounds-rule}}
\label{sec:ctrex}

\citet{ChatterjeeGMZ22} claimed that their \cref{proofrule:SI-simple-lower-bounds-rule} is complete for all programs.
As promised, we now demonstrate a counterexample to their claim of completeness.

\begin{figure}%[t]{r}{0.3\linewidth}
    \input{chatterjee-rule-exception.tikz}
    \caption{Counterexample to the SI-rule for lower bounds. \textnormal{With an initial state of $(l_1, (1, 2, 1/4))$, the termination probability of this program is $1/2$. 
However, there is no SI that shows this.}}
    \label{fig:chatterjee-rule-exception}
\end{figure}

\cref{proofrule:SI-simple-lower-bounds-rule} for lower bounds can be applied onto any $\iCFG$ that induces a set of states $\SIpred$ with the property that executions remain within $\SIpred$ with exactly the probability of termination. 
As mentioned previously, not all programs are so well behaved. 
Consider the program $\cK$ defined in \cref{fig:chatterjee-rule-exception}.

% SATHIYA: A comment on the counterexample to Chatterjee's rule: 
% I was able to come up with this simple example by modeling the probability values as fractional expressions over the variables. 
% But this program cannot be easily translated into Chatterjee's original definitions for the CFGs, as he doesn't (directly) allow for probability distributions to be parametrized by expressions over the program variables. 
%
% Since this section purports to give a counterexample, a more directly translatable program seems desirable. 
% I have a simpler counterexample, that uses the following value as the probability of non-termination: 
% $$
% \prod_{i = 1}^\infty \left( 1 - \frac{1}{2^n} \right) 
% $$
% It's easy to build an $\iCFG$ for this program. However, this probability value is a bit hard to understand. 
% It's actually the [Euler function](https://en.wikipedia.org/wiki/Euler_function) with parameter 1/2, which doesn't appear to have a simple closed form. 
% [Here's](https://math.stackexchange.com/questions/1200575/what-is-the-value-of-prod-i-1-infty-left1-frac12i-right) a proof of a lower bound for this function at $e^{-3/2}$. 
%
% I'm not sure which example to present here; for now (and because it lets me cite the buffon machines paper), I'll use the simpler example. 
% If I come up with a better example, I'll substitute here. 

The initial location of $\cK$ is $l_1$, and the values of the variables $(x_1, x_2, x_3)$ are $(1, 2, 1/4)$. $l_1$ is a probabilistic location, $l_3$ is an assignment location, and $l_2$ is a terminal location. 
It isn't difficult to prove that the probability of termination of $\cK$ is $1/2$; we leave the details to the diligent reader. 
The SI-rule for lower bounds requires a stochastic invariant $(\SIpred_\cK, 1/2)$ such that executions almost-surely either terminate or exit $\SIpred_\cK$. 

\begin{lemma}
    There is no stochastic invariant $(\SIpred_\cK, 1/2)$ of $\cK$ such that runs almost-surely either terminate or leave $\SIpred_\cK$. 
\end{lemma}
\begin{proof} % This can be shortened, or even removed. 
    Suppose there does exist a stochastic invariant $(\SIpred_\cK, 1/2)$ that satisfies these properties. Therefore, the probability measure of the union of the collection of runs $\sfLeave_{\SIpred}$ that leave $\SIpred_\cK$ and the runs $\sfTerm_{\SIpred}$ that terminate inside $\SIpred_\cK$ is $1$. However, because $\SIpred_\cK$ is a stochastic invariant, the probability measure of $\sfLeave_{\SIpred}$ is bounded above by $1/2$. This means that the measure of $\sfTerm_{\SIpred}$ is bounded below by $1/2$. But, the termination probability of $\cK$ is $1/2$. Consequently, the measure of $\sfTerm_{\SIpred}$ must be exactly $1/2$. This means $\sfTerm_{\SIpred}$ contains all terminating runs of $\cK$. 

    It is easy to see that from any state $(l, \bfx)$ reachable from the initial state, there is a finite path of length at most $2$ that leads it to a terminal state. Therefore, all reachable states $(l, \bfx)$ are a part of some terminating run; meaning that the set of states that make up the runs in $\sfTerm_{\SIpred}$ must be the set of reachable states. This is only possible when $\SIpred_\cK$ is the set of reachable states. This means no runs leave $\SIpred_\cK$, and therefore, the measure of $\sfTerm_{\SIpred}$ is $1$. This contradicts the fact that the termination probability of $\cK$ is $1/2$. 
\end{proof}

\emph{A note on syntax.}
The $\CFG$ of \citet{ChatterjeeGMZ22} over which the claim of completeness of \cref{proofrule:SI-simple-lower-bounds-rule} was made do not feature fractional expressions guiding probabilistic branching. 
Nevertheless, they can be simulated with small programs that only use the basic coin flip \cite{FlajoletPS11}. 
Therefore, \cref{fig:chatterjee-rule-exception} is a valid counterexample to their claim. 

% \VRS{To be honest, I'm not quite sure this is necessary.}
\subsection{Applications to Qualitative Termination}
\label{subsec:ast-si}
Notice that by setting $p = 0$, our \cref{proofrule:lower-bounds-ours} for proving lower bounds on termination probabilities can be used to prove the almost-sure termination of a program.
This is a more general consequence of the unrolling lemma, that if a program terminates with probability $1$, it only explores a finite amount of its state space to amass a termination probability of $1 - \epsilon$ for every $\epsilon > 0$.
However, $\AST$ is a property that holds when the initial state is changed to any reachable state.
If $\cG$ is almost surely terminating, then $\cG$ is almost surely terminating from any state $\sigma$ reachable from the initial state $\sigma_0 = (l_{\init}, \bfx_{\init})$.
Furthermore, if from every reachable state $\sigma$, $\cG$ is known to terminate with some minimum probability $p > 0$, then the program is $\AST$.
We can exploit these facts to design another rule, based on stochastic invariants, to prove the almost-sure termination of a program.

\begin{proofrule}[SI-Indicators for $\AST$]
    \label{proofrule:ast-SI}
    To prove that $\cG$ is $\AST$, fix a $0 < p < 1$ and find
    \begin{enumerate}
        \item an inductive invariant $\Inv$ containing the initial state,
        \item a mapping $\SI$ from each $\sigma \in \Inv$ to SI-indicator functions $(\SI_\sigma, p) : \Inv \to \setOfReals$ such that $\SI_\sigma(l, \bfx) \leq p$ and, for all $(l, \bfx) \in \Inv$,
        \begin{enumerate}
            \item $\SI_\sigma(l, \bfx) \geq 0$.
            \item if $(l, \bfx)$ is an assignment, or nondeterministic state, then $\SI_\sigma(l, \bfx) \geq \SI_\sigma(l', \bfx')$ for every successor $(l', \bfx')$.
            \item if $(l, \bfx)$ is a probabilistic state, then $\SI_\sigma(l, \bfx) \geq \sum \sfPr((l, l'))[\bfx] \times \SI_\sigma(l', \bfx')$ over all possible successor states $(l', \bfx')$.
        \end{enumerate}
        \item a mapping $\cE$ mapping states $\sigma \in \Inv$ to values $\epsilon_\sigma \in (0, 1]$,
        \item a mapping $\cH$ mapping states $\sigma \in \Inv$ to values $H_\sigma \in \setOfNaturals$,
        \item a mapping $\cU$ from each $\sigma \in \Inv$ to variants $U_\sigma : \Inv \to \setOfNaturals$ that is bounded above by $\cH(\sigma)$, maps all states $\{\gamma \mid \SI_\sigma(\gamma) \geq 1 \lor \gamma \text{ is terminal} \}$ to $0$ and, for other states $(l, \bfx) \in \Inv$,
        \begin{enumerate}
            \item if $(l, \bfx)$ is an assignment, or nondeterministic state, $U_\sigma(l', \bfx') < U_\sigma(l, \bfx)$ for every successor $(l', \bfx')$.
            \item if $(l, \bfx)$ is a probabilistic state, $\sum \sfPr(l, l')[\bfx] > \cE(\sigma)$ over all successor states $(l', \bfx')$ with $U_\sigma(l', \bfx') < U_\sigma(l, \bfx)$.
        \end{enumerate}
    \end{enumerate}
\end{proofrule}
Intuitively, our rule requires SI-indicator functions $\SI_\sigma$ at each $\sigma \in \Inv$ that hold for executions beginning at $\sigma$ with probability $\geq 1 - p$.
Each of these imply stochastic invariants $(\SIpred_\sigma, p)$ centered around the state $\sigma$.
The functions $\cU$, $\cH$, and $\cE$ combine together to form variant functions $U_\sigma$ of the McIver-Morgan kind at each $\sigma \in \Inv$.
These $U_\sigma$ further imply that a terminal state is contained within each $\SIpred_\sigma$, and induces paths within each $\SIpred_\sigma$ to this terminal state.
Feeding $U_\sigma$, $\epsilon_\sigma$, and $H_\sigma$ into McIver and Morgan's variant \cref{proofrule:ast-variant} gives us a proof for the fact that, were the execution to be restricted to $\SIpred_\sigma$, the probability of termination from $\sigma$ is $1$.
Therefore, the probability of termination from each $\sigma$ is $\geq 1 - p$.
Applying the zero-one law of probabilistic processes \cite[Lemma 2.6.1]{McIverMorganBook} completes the proof of soundness of this rule.

Notice that we don't mandate any locality conditions on the SI-indicators in this rule.
This is because they aren't necessary to infer the soundness of the rule.
However, we show in our completeness proof that one can always find ``local'' SI-indicators that only take on values $< 1$ at finitely many states for $\AST$ programs.
This is because these SI-indicators are built from appropriate finite stochastic invariants, the existence of which is a consequence of the unrolling lemma.

\begin{lemma}[Completeness]
    \label{lem:SI-ast-rule-completeness}
    \cref{proofrule:ast-SI} is relatively complete. 
\end{lemma}
The proof for this lemma is principally similar to the proof of \cref{lem:lower-bounds-completeness}, and we omit it here for reasons of space.

%% file: figures/loop-exploitation.tikz
\begin{tikzpicture}
	\begin{pgfonlayer}{nodelayer}
		\node [style=plain circle] (0) at (0, 0) {};
		\node [style=blue circle] (1) at (2, 1) {};
		\node [style=blue-green circle] (2) at (4, -0.5) {};
		\node [style=blue circle] (3) at (4, 2) {};
		\node [style=none] (4) at (2.5, 1.25) {};
		\node [style=none] (5) at (0.5, 0.25) {};
		\node [style=none] (6) at (2, -0.25) {};
		\node [style=black circle] (7) at (5, 2.5) {};
		\node [style=none] (8) at (6, -0.75) {};
		\node [style=blue-green circle] (9) at (8, -1) {};
		\node [style=black circle] (10) at (10, -1.25) {};
		\node [style=none] (11) at (-1.5, 0) {};
		\node [style=none] (12) at (5, 1) {};
		\node [style=none] (13) at (7, 0.5) {};
		\node [style=none] (14) at (0.25, -0.75) {{\scriptsize $\sigma \in \Sigma_{good}$ }};
	\end{pgfonlayer}
	\begin{pgfonlayer}{edgelayer}
		\draw [style=dashed directional edge] (5.center) to (1);
		\draw (0) to (5.center);
		\draw (0) to (6.center);
		\draw (1) to (4.center);
		\draw [style=dashed directional edge] (6.center) to (2);
		\draw [style=dashed directional edge] (4.center) to (3);
		\draw [style=direction edge] (3) to (7);
		\draw [style=dashed directional edge] (8.center) to (9);
		\draw (2) to (8.center);
		\draw [style=direction edge] (9) to (10);
		\draw [style=red full edge, bend right=300] (9) to (2);
		\draw [style=red full edge, bend left=300, looseness=1.25] (3) to (1);
		\draw [style=direction edge] (11.center) to (0);
		\draw [style=dashed directional edge] (0) to (12.center);
		\draw [style=dashed directional edge] (0) to (13.center);
	\end{pgfonlayer}
\end{tikzpicture}

%% file: figures/chatterjee-rule-exception.tikz
\begin{tikzpicture}
	\begin{pgfonlayer}{nodelayer}
		\node [style=plain circle] (0) at (0, 0) {{\scriptsize $l_1$}};
		\node [style=plain circle] (1) at (4, 0) {{\scriptsize $l_2$}};
		\node [style=plain circle] (2) at (-4, 0) {{\scriptsize $l_3$}};
		\node [style=none] (3) at (-2, 1.25) {{\scriptsize $x_1/x_2$}};
		\node [style=none] (4) at (-2, -1.75) {{\scriptsize $\begin{matrix} x_1 \coloneqq x_2 \\ x_2 \coloneqq x_2 + x_3 \\ x_3 \coloneqq x_3 / 2 \end{matrix}$}};
		\node [style=none] (5) at (2, 0.5) {{\scriptsize $1 - (x_1/x_2)$}};
		\node [style=none] (6) at (0, 1.25) {};
	\end{pgfonlayer}
	\begin{pgfonlayer}{edgelayer}
		\draw [style=direction edge] (0) to (1);
		\draw [style=direction edge, bend right] (0) to (2);
		\draw [style=direction edge, bend right, looseness=1.25] (2) to (0);
		\draw [style=direction edge] (6.center) to (0);
	\end{pgfonlayer}
\end{tikzpicture}

%% file: applications.tex
\section{Discussion: Traveling Between Proof Systems}
\label{sec:applications}

We have presented sound and relatively complete proof rules for qualitative and quantitative termination of probabilistic programs with bounded probabilistic and nondeterministic choice.
We have also demonstrated relative completeness of our rules in the assertion language of arithmetic.
% Our proof rules combine the familiar ingredients of supermartingales and variant functions in novel ways to reach completeness.
However, a new proof rule is ultimately interesting only if one can actually prove the termination of many programs.
In order to show that our proof rules, in addition to their theoretical properties, are also applicable in a variety of situations,
we demonstrate that proofs in many existing proof systems can be compiled into our proof rules. 
 
% \VRS{TODO: Add mention of the guards of \citet{FengCSKKZ23}, the new proof rule of \citet{McIverMKK18}, and so on.}

\subsubsection*{From \citet{McIverMorganBook}}
The variant functions from \cref{proofrule:ast-variant} immediately form the variant functions required in \cref{proofrule:ast-martingale}.
Set $V$ to $0$ at the terminal state and $1$ everywhere else.
This gives all we need to apply \cref{proofrule:ast-martingale}.

%\subsubsection*{From McIver-Morgan-Kaminski-Katoen \cite{McIverMKK18}}
\subsubsection*{From \citet{McIverMKK18}}

The $\AST$ \cref{proofrule:ast-mciver-katoen} proposed by \citet{McIverMKK18} has been applied onto a variety of programs, and has been shown to be theoretically applicable over the 2D random walker.
% Applications of their rule effectively requires the construction of a distance variant that is also a supermartingale.
We note that their variants can be reused in \cref{proofrule:ast-martingale} with little alterations as both the supermartingale and variant functions.
This means proofs in their rule can be easily translated to proofs that use \cref{proofrule:ast-martingale}.
For the reverse direction, the construction of the supermartingale-variant we used to demonstrate the completeness of their rule (see \cref{sec:completeness-mciver-katoen-rule}) induces a technique to translate proofs using our rule to ones using theirs.
% For the reverse direction, the construction of the supermartingale-variant we used to demonstrate the completeness of their rule (see our full version) induces a technique to translate proofs using our rule to ones using theirs.

\subsubsection*{From \cref{proofrule:ast-martingale} to \cref{proofrule:ast-SI}}
\newcommand{\tempVUp}{v_{\relshortuparrow}}

Hidden in the proof of the soundness of \cref{proofrule:ast-martingale} are the stochastic invariants that form the basis of \cref{proofrule:ast-SI}.
Fix a $p$, and take the set $\SIpred_\sigma = \{ \gamma \in \Inv \mid V(\gamma) \leq \tempVUp \}$ for a sufficiently high value of $\tempVUp$ to yield an upper bound of $p$ on the probability of exiting $\SIpred_\sigma$.
Then, expand $\SIpred_\sigma$ with the states necessary to keep all shortest consistent terminal runs across schedulers from states in $\SIpred_\sigma$ entirely within $\SIpred_\sigma$.
In spite of these extensions, the value of $V$ will be entirely bounded when restricted to states in $\SIpred_\sigma$.
It is then trivial to build the indicator functions from each stochastic invariant $(\SIpred_\sigma, p)$ and the variant functions from $U$, completing a translation from \cref{proofrule:ast-martingale} to \cref{proofrule:ast-SI}.
Note that using this technique, one can translate proofs from \citet{McIverMKK18} and \citet{McIverMorganBook} to \cref{proofrule:ast-SI} as well.

\subsubsection*{Using Stochastic Invariants \cite{ChatterjeeGMZ22}}
Separately, \citet{ChatterjeeGMZ22} have shown the applicability of \cref{proofrule:SI-simple-lower-bounds-rule} to demonstrate lower bounds on the termination probabilities for a variety of programs, and have also presented template-based synthesis techniques for achieving limited completeness.
These proofs are also valid for \cref{proofrule:lower-bounds-ours}, by setting the same stochastic invariant for each $n$.
% We believe their techniques can be translated to produce proofs that use \cref{proofrule:lower-bounds-ours}. % I think this is possible, but I'm not sure if we should say this here.
% \VRS{The final sentence needs a revision after looking at Chatterjee's template synthesis algorithms.}

%% file: appendix.tex
\section{Completeness of the rule of McIver et al. \cite{McIverMKK18}}
\label{sec:completeness-mciver-katoen-rule}

In this section, we prove that \cref{proofrule:ast-mciver-katoen} of \citet{McIverMKK18} is complete.

Take an $\AST$ $\CFG$ $\cG$, and denote its set of reachable states by $\reach(\cG)$.
Our starting point is the supermartingale $V : \reach(\cG) \to \setOfReals$ that we constructed in \cref{sec:ast-martingale}.
We will modify this $V$ in a few subtle ways to match \cref{proofrule:ast-mciver-katoen}.
Recall from our prior construction that the sublevel sets $V^{\leq r} = \{\sigma \mid V(\sigma) \leq r\}$ are finite for this $V$.
It is easy to see that if this $V$ were one-one, it is easy to produce progress functions $p$ and $d$ for $V$ that match the requirements of \cref{proofrule:ast-mciver-katoen}.
This is due the finite nature of the sublevel sets.

We will now work toward an injective supermartingale function that yields finite sublevel sets.
We do this by first injecting ``gaps'' into $V$ without affecting its supermartingale status.
This makes it a \emph{strict supermartingale}, which is a supermartingale that strictly reduces in expectation at certain states.
We also say that a supermartingale $V$ is \emph{strict at a state $\sigma$} if, in a single step from $\sigma$, executions are expected to reach some state $\sigma'$ with $V(\sigma') < V(\sigma)$.

After procuring our desired strict supermartingale, we will mildly perturb the values assigned to its level sets in a way that doesn't affect its supermartingale properties.
This perturbation gives us an injective supermartingale.
Notice that this kind of perturbation is only possible because of the gaps; perturbation typically invalidates martingale properties.

We start with the function $V_0 : \reach(\cG) \to \setOfReals$, defined as
\begin{equation}
    \label{eq:V_0}
    V_0(\sigma) \triangleq \ln(1 + V(\sigma))
\end{equation}

\begin{lemma}
    $V_0$ is a supermartingale that assigns $0$ to the terminal state, induces finite sublevel sets, and is strict at all probabilistic states $\sigma$ from which there is a transition in $\cG$ to a successor $\sigma'$ such that $V(\sigma') \neq V(\sigma)$.
\end{lemma}
\begin{proof}
    It is easy to see that the function $V'$ that maps $\sigma \mapsto V(\sigma) + 1$ is also a supermartingale.
    Thus, $V_0(\sigma) = \ln(V'(\sigma))$.
    It's again easy to see why $V_0$ satisfies the supermartingale conditions at assignment and nondeterministic states.
    We will now demonstrate that $V_0$ is a supermartingale at probabilistic states $(l, \bfx)$.
    Let $\{(l_1, \bfx), (l_2, \bfx), \ldots (l_n, \bfx) \}$ be the collection of successor states of $(l, \bfx)$.
    Then, take
    $$
    \sum_{i = 1}^n \sfPr(l, l_i) V_0(l_i, \bfx) = \sum_{i = 1}^n \sfPr(l, l_i) \ln(V'(l_i, \bfx)) = \ln \left(\prod_{i=1}^n \left(V'(l_i, \bfx)\right)^{\sfPr(l, l_i)} \right)
    $$
    Where $\sfPr$ is the transition probability function of $\cG$.

    The weighted AM-GM inequality, together with the supermartingale nature of $V'$, implies that
    $$
    \sum_{i=1}^n \sfPr(l, l_i) V'(l_i, \bfx) \geq \prod_{i=1}^n \left(V'(l_i, \bfx)\right)^{\sfPr(l, l_i)} 
    \implies
    V_0(l, \bfx) \geq \sum_{i = 1}^n \sfPr(l, l_i) V_0(l_i, \bfx)
    $$
    Thus proving the supermartingale nature of $V_0$.

    Observe that $V(\sigma_\bot) = 0 \implies V_0(\sigma_\bot) = 0$ at the terminal state $\sigma_\bot$.
    Furthermore, the sublevel sets $V_0^{\leq r} = \{ \sigma \in \reach(\cG) \mid V_0(\sigma) \leq r \}$ are finite for each $r$.
    This is because the corresponding sublevel sets $V^{\leq r}$ are finite and the function $x \mapsto \ln(1 + x)$ is strictly increasing for increasing $x$.

    Now, suppose from $(l', \bfx')$, there is a transition to $(l'_i, \bfx')$ such that $V(l', \bfx') \neq V(l'_i, \bfx')$.
    The weighted AM-GM inequality implies that at such states,
    $$
    \sum_{i=1}^n \sfPr(l', l'_i) V'(l'_i, \bfx') > \prod_{i=1}^n \left(V'(l'_i, \bfx')\right)^{\sfPr(l', l'_i)} \implies V_0(l', \bfx') > \sum_{i = 1}^n \sfPr(l', l'_i) V_0(l'_i, \bfx')
    $$
    Where $(l'_1, \bfx'), \ldots (l'_n, \bfx')$ are successors of $(l', \bfx')$.
    Hence, $V_0$ is strict at $(l', \bfx')$.
    % This completes the proof.
\end{proof}

$V_0$ is thus a supermartingale with ``gaps'' at probabilistic states $\sigma$ from which the value of $V$ has a non-zero probability of reduction in a single step.
These gaps induce some ``wiggle room'' for the values assigned by $V_0$ at these $\sigma$ within which the supermartingale nature of $V_0$ is unaffected.
One could arbitrarily pick one such $\sigma$ and subtract a small $\epsilon > 0$ from $V_0(\sigma)$ and still have a supermartingale.

\begin{lemma}%[Wiggle room lemma?]
    \label{lem:supermartingale-perturbation}
    Let $U$ be a non-negative supermartingale function over the reachable space $\reach(\cG)$ of an $\AST$ $\CFG$ $\cG$ that is $0$ only at the terminal state, and induces finite sublevel sets.

    For each $r > 0$ in the range of $U$, there is a state $\sigma_r \in \reach(\cG)$ with $U(\sigma_r) = r$ and a real $\epsilon_r > 0$ such that the function $U_r : \reach(\cG) \to \setOfReals$ defined below is a non-negative supermartingale that induces finite sublevel states and is $0$ only at the terminal state.
    $$
    U_r(\gamma) = \begin{cases}
        U(\sigma_r) - \epsilon_r & \gamma = \sigma_r \\
        U(\gamma) & \gamma \neq \sigma_r
    \end{cases}
    $$
\end{lemma}
\begin{proof}
    Take the level set $U^{= r} = \{\sigma \in \reach(\cG) \mid U(\sigma) = r \}$.
    The set $U^{= r}$ is non-empty, finite, and doesn't contain the terminal state.
    This is because $r$ is in the range of $U$, $U$ induces finite sublevel sets, and $r > 0$.

    We first show that there must exist a state $\sigma \in U^{=r}$ such that $U$ is strict at $\sigma$.
    Suppose no such state exists.
    Then, $U$ must be an exact martingale at all states in $U^{=r}$.
    This means all successors of each $\sigma \in U^{=r}$ must be assigned the value $r$ by $U$.
    This contradicts $\cG \in \AST$, as when executions enter $U^{=r}$, they stay inside $U^{=r}$ with probability $1$ and never reach a terminal state.
    Note that there will exist executions that enter $U^{=r}$ because each state in $U^{=r}$ is reachable with non-zero probability.
    Pick $\sigma_r$ to be this state in $U^{=r}$.
    Thus, $U$ is strict at $\sigma_r$.

    Consider the collection of supermartingale inequalities $\cE_r$ of $U$ that involve $\sigma_r$.
    Inequalities in $\cE_r$ are derived from transitions where $\sigma_r$ is either the target or the source.
    Since $U$ is strict at $\sigma_r$, the inequalities from transitions sourced at $\sigma_r$ are strict.
    The inequalities from transitions that target $\sigma_r$ may not be strict.

    If $\sigma_r$ is either assignment or nondeterministic, the inequalities from transitions sourced at $\sigma_r$ take the form $U(\sigma_r) > U(\sigma')$ for successors $\sigma'$.
    If it is probabilistic, the single inequality sourced at $\sigma_r$ takes the form $U(\sigma_r) > \sum p_i U(\sigma_i)$ for successors $\sigma_i$ reached with probability $p_i$.
    Either way, there must exist a real $\epsilon_r > 0$ smaller than the difference between $U(\sigma_r)$ and the RHS of these inequalities, i.e., either all of the $U(\sigma')$ or the sum $\sum p_i U(\sigma_i)$.
    Updating $U(\sigma_r)$ by subtracting $\epsilon_r$ from it invalidates none of these strict inequalities.

    Inequalities in $\cE_r$ derived from transitions that target $\sigma_r$ place $U(\sigma_r)$ on the right of a $\geq$ relation.
    Clearly, reducing $U(\sigma_r)$ will not affect the validity of these inequalities.
    Therefore, this subtraction operation doesn't affect the validity of $\cE_r$, meaning that it doesn't change the supermartingale nature of $U$.
    It is also easy to see that despite this change, $U$ still assigns $0$ to the terminal states and induces finite sublevel sets.
    This completes the proof.
\end{proof}

We will now describe an infinite sequence of supermartingales $(V_n)_{n \in \setOfNaturals}$ that are increasingly injective.
Each of these supermartingales assigns $0$ to the terminal state, and induces finite sublevel sets.
Each element $V_{n + 1}$ in the sequence resolves a single two-to-one mapping in the previous element $V_n$.
This clash is resolved by reducing the value of one of two states involved in the clash by a suitable tiny $\epsilon$ in the style of the above \cref{lem:supermartingale-perturbation}.
The limit $V_\infty$ of the sequence will be completely injective.

The first element in this sequence is the function $V_0$ we defined in \cref{eq:V_0}.
We now describe a procedure to define $V_{n + 1}$ from $V_n$.
In $V_n$, pick the smallest value $r > 0$ such that the cardinality of the level set $V^{=r}_n = \{\sigma \in \reach(\cG) \mid V_n(\sigma) = r \}$ is at least $2$.
\cref{lem:supermartingale-perturbation} indicates the presence of one state $\sigma_r \in V^{=r}_n$ and an $\epsilon_r > 0$ such that the function $V'_n$ that is identical to $V_n$ except at $\sigma_r$ where $V'_n(\sigma_r) = V_n(\sigma_r) - \epsilon_r$ is still a supermartingale with all the properties we desire.
Let $r'$ be the largest number under $r$ such that the level set $V^{=r'}_n$ at $r'$ is non empty.
Reduce $\epsilon_r$ so that $V_n(\sigma_r) - \epsilon_r > r'$.
Define $V_{n + 1}$ now as
$$
V_{n+1}(\gamma) = \begin{cases}
    V_n(\sigma_r) - \epsilon_r & \gamma = \sigma_r  \\
    V_n(\gamma) & \gamma \neq \sigma_r
\end{cases}
$$
Thus, the level set $V^{=r}_{n+1}$ is strictly smaller than $V^{=r}_{n}$.
Furthermore, both $V_n$ and $V_{n + 1}$ map all states in the sublevel set $V^{<r}_r = \{\sigma \in \reach(\cG) \mid V_n(\sigma) < r\}$ to different values.
In other words, they are both injective when restricted to $V^{< r}_n$.

The finiteness of the sublevel sets of $V_0$ implies that for each $r > 0$, there is an $m_r$ such that the sublevel set $V^{\leq r}_{m_r}$ is free of clashes.
In other words, for each $r$, there is a $m_r$ such that $V_{m_r}$ is injective when restricted to the sublevel set $V^{\leq r}_{m_r}$.
Thus, the limit supermartingale $V_\infty$ defined as
$$
V_\infty(\sigma) = \lim_{n \to \infty} V_n(\sigma)
$$
is injective, induces finite sublevel sets, and assigns $0$ to the terminal state.

\begin{theorem}
    \cref{proofrule:ast-mciver-katoen} is relatively complete.
\end{theorem}
\begin{proof}
    Set the inductive invariant $\Inv$ to $\reach(\cG)$.
    The supermartingale $V_\infty$ defined above injective, induces finite sublevel sets, and assigns $0$ to the terminal state.

    The range of $V_\infty$ is countable because its domain is countable.
    Order the elements of the range of $V_\infty$ in an increasing sequence $(r_n)_{n \in \setOfNaturals} = r_1 < r_2 < \ldots$.
    Index the states in $\reach(\cG)$ in such a way that the state indexed $n$ is assigned the value $r_n$.
    We will refer to the $n^{th}$ state in this index by $\sigma_n$.
    Clearly, the sublevel set at any $r_n$ is $V^{\leq r_n}_\infty = \{\sigma \in \Inv \mid V_\infty(\sigma) \leq r_n\} = \{\sigma_i \in \Inv \mid i \leq n\}$.

    We now define the two antitone functions $p$ and $d$ that \cref{proofrule:ast-mciver-katoen} demands.
    Fix an input real $r > 0$.
    Let $n \in \setOfNaturals$ be the largest number so that $r_n \leq r$.
    We define $d(r)$ to be the smallest difference between $r_{i + 1}$ and $r_i$ for all $i \leq n$.
    We then define $p(r)$ to be the smallest probability of a transition that links a state $\sigma_i$ with $i \leq n$ to a state $\sigma_j$ with $j < i$.
    We will use $\sigma_i \rightarrow_p \sigma_j$ to indicate the presence of a transition in $\cG$ that links $\sigma_i$ to $\sigma_j$ with probability $p$.
    Notice that, since $V_\infty$ is injective, and since $\cG$ is $\AST$, this $\sigma_j$ must exist for each $\sigma_i$.
    Thus,
    \begin{equation*}
        p(r) \triangleq \inf_{i \leq n} \left\{ p \in (0, 1) \mid \sigma_i \rightarrow_{p} \sigma_j \land j < i \right\} \quad\quad\quad\quad d(r) \triangleq \min_{i < n}\left(r_{i + 1} - r_i\right)
    \end{equation*}
    It's easy to see from the definitions that $p$ and $d$ are antitone.
    Furthermore, from their definitions, at each state $\sigma_i$ with $i \leq n$, there is a transition with probability $\geq p(r)$ to a state $\sigma_j$ with $j < i$, and therefore with $V_\infty(\sigma_i) - V_\infty(\sigma_j) \geq d(r)$.
    Thus, $V_\infty$ matches the progress conditions induced by $p$ and $d$.

    %TODO: This is a nearly pointless paragraph.
    We now discuss how each of $\Inv$, $V_\infty$, $p$ and $d$ can be represented in our assertion language.
    The representation of $\Inv$ and the base supermartingale $V$ from \cref{sec:ast-martingale} has been discussed prior.
    The function $V_0$ can be represented in our assertion language through the Taylor expansion of the natural logarithm function.
    Each function $V_n$ in the sequence can be represented, as the definition of $V_{i + 1}$ from $V_i$ is a representable transformation.
    Thus, the limit function $V_\infty$ can be represented in our assertion language.
    Finally, each progress function $p$ and $d$ can be represented through the representation of $V_\infty$.
    This completes the proof.
\end{proof}

\section{Applying the Rule of McIver et al. \cite{McIverMKK18} for the 2D Random Walker}
\label{sec:mciver-katoen-2drw}

As we mentioned in \cref{subsec:mciver-katoen-rule}, the supermartingale-variant function $V$, defined as
$$
V(x, y) \triangleq \sqrt{\ln \left(1 + ||x,y|| \right)}
$$
Where $||x,y||$ stands for the Euclidean distance $\sqrt{x^2 + y^2}$ of the point $(x, y)$ from the origin $(0, 0)$, is a certificate for the application of \cref{proofrule:ast-mciver-katoen} for the random walker on the 2 dimensional grid.
It's clear to observe that $p(v) = 1/4$ is a sufficient antitone progress function for the probability of reduction, as at least one of the 4 directions available to the walker must reduce its Euclidean distance from the origin.
We now show that the function $d$ defined below completes the application of their rule for the 2D random walker.
$$
d(v) = \frac{1}{2v} \left( v^2 - \ln\left(e^{v^2} - 1\right) \right)
$$
First, for any value $v > 0$, observe that there must be a circle of points $(x, y)$ on the 2-D plane such that $V(x, y) = v$.
It's easy to show that the function $V$ must reduce along one direction of the four directions by the least amount when $x = y$.
Hence, for a value of $v > 0$, pick $x_v > 0$ such $V(x_v, x_v) = v$.
This means the minimum reduction of $v$, which we refer to as $\Delta_v$, is
\begin{align*}
    \Delta_v &= V(x_v, x_v) - V(x_v - 1, x_v)\\
    &= \sqrt{\ln\left(1+\sqrt{2} x_v\right)} - \sqrt{\ln\left(1+\sqrt{(x-1)^2 + x^2}\right)}\\
    &> \sqrt{\ln\left(1+\sqrt{2} x_v\right)} - \sqrt{\ln\left(1+\sqrt{2}\left(x_v - \frac{1}{\sqrt{2}}\right)\right)}\\
    \implies \Delta_v &> \sqrt{\ln\left(1+\sqrt{2} x_v\right)} - \sqrt{\ln\left(\sqrt{2} x_v\right)}
\end{align*}
But,
\begin{align*}
    \sqrt{\ln\left(1+\sqrt{2} x_v\right)} - \sqrt{\ln\left(\sqrt{2} x_v\right)} = \frac{1}{\sqrt{\ln\left(1+\sqrt{2}x_v\right)} + \sqrt{\ln\left(\sqrt{2}x_v\right)}} \ln\left(1+\frac{1}{\sqrt{2}x_v}\right)
\end{align*}
Given that
$$
v = \ln\left(1+\sqrt{2} x_v\right) \implies x = \frac{e^{v^2} - 1}{\sqrt{2}} \quad\quad\text{and}\quad\quad \sqrt{\ln\left(1+\sqrt{2}x_v\right)} > \sqrt{\ln\left(\sqrt{2}x_v\right)}
$$
We have
\begin{align*}
    \Delta_v &> \frac{1}{2v} \ln\left(1+\frac{1}{e^{v^2} - 1}\right)\\
    \implies \Delta_v &> \frac{1}{2v}\left(v^2 - \ln\left(e^{v^2} - 1\right) \right)
\end{align*}
Meaning $\Delta_v > d(v)$, as desired.